\documentclass[11pt]{article}
\usepackage{amsmath,amssymb,amsthm,xspace,fullpage}
\usepackage[usenames,dvipsnames]{color}
\usepackage[pagebackref,letterpaper=true,colorlinks=true,pdfpagemode=none,urlcolor=blue,linkcolor=blue,citecolor=BrickRed,pdfstartview=FitH]{hyperref}

\newtheorem{theorem}{Theorem}[section]
\newtheorem{lemma}[theorem]{Lemma}
\newtheorem{definition}[theorem]{Definition}
\newtheorem{proposition}[theorem]{Proposition}
\newtheorem{claim}[theorem]{Claim}
\newtheorem{corollary}[theorem]{Corollary}

\newcommand{\bbC}{\mathbb{C}}
\newcommand{\bbF}{\mathbb{F}}
\newcommand{\bbN}{\mathbb{N}}
\newcommand{\bbR}{\mathbb{R}}
\newcommand{\bbT}{\mathbb{T}}
\newcommand{\bbU}{\mathbb{U}}
\newcommand{\bbZ}{\mathbb{Z}}
\newcommand{\bfa}{\mathbf{a}}
\newcommand{\bfb}{\mathbf{b}}
\newcommand{\bfd}{\mathbf{d}}

\newcommand{\bfh}{\mathbf{h}}

\newcommand{\bfx}{\mathbf{x}}
\newcommand{\bfL}{\mathbf{L}}
\newcommand{\bfP}{\mathbf{P}}
\newcommand{\bfQ}{\mathbf{Q}}
\newcommand{\bfR}{\mathbf{R}}
\newcommand{\bfS}{\mathbf{S}}
\newcommand{\bit}{\{0,1\}}
\newcommand{\Fp}{\mathbb{F}_p}

\newcommand{\calB}{\mathcal{B}}

\newcommand{\calI}{\mathcal{I}}

\newcommand{\calP}{\mathcal{P}}
\newcommand{\calR}{\mathcal{R}}

\newcommand{\calV}{\mathcal{V}}
\newcommand{\sfe}{\textsf{e}}

\newcommand{\rank}{\mathrm{rank}}
\newcommand{\depth}{\mathrm{depth}}
\newcommand{\E}{\mathop{\mathbf{E}}}
\newcommand{\Var}{\mathop{\mathbf{Var}}}
\newcommand{\dtv}{\mathrm{d_{TV}}}

\begin{document}
\title{A Characterization of Locally Testable Affine-Invariant Properties via Decomposition Theorems}
\author{
  Yuichi Yoshida\thanks{
    Supported by JSPS Grant-in-Aid for Research Activity Start-up (24800082), MEXT Grant-in-Aid for Scientific Research on Innovative Areas (24106001), and JST, ERATO, Kawarabayashi Large Graph Project.
  }\\
  National Institute of Informatics and Preferred Infrastructure, Inc.\\
  \texttt{yyoshida@nii.ac.jp}
}

\maketitle

\begin{abstract}
  Let $\calP$ be a property of function $\Fp^n \to \bit$ for a fixed prime $p$.
  An algorithm is called a tester for $\calP$ if, 
  given a query access to the input function $f$, 
  with high probability,
  it accepts when $f$ satisfies $\calP$ and rejects when $f$ is ``far'' from satisfying $\calP$.
  In this paper, 
  we give a characterization of affine-invariant properties that are (two-sided error) testable with a constant number of queries.
  The characterization is stated in terms of decomposition theorems, 
  which roughly claim that any function can be decomposed into a structured part that is a function of a constant number of polynomials, and a pseudo-random part whose Gowers norm is small.
  We first give an algorithm that tests whether the structured part of the input function has a specific form.
  Then we show that an affine-invariant property is testable with a constant number of queries if and only if 
  it can be reduced to the problem of testing whether the structured part of the input function is close to one of a constant number of candidates.
\end{abstract}

\section{Introduction}
In \emph{property testing}, 
we want to distinguish objects that satisfy a predetermined property $\calP$ from objects that are ``far'' from satisfying $\calP$.
Intuitively, we say that an object is far from satisfying $\calP$ if we must modify a constant fraction of the object to make it satisfy $\calP$.
By ignoring objects that do not satisfy $\calP$ but are close to satisfying $\calP$,
sometimes we can design very efficient algorithms for testing $\calP$ that run even in constant time, which is independent of the object size.
For an overview of recent developments in this area, we refer the reader to surveys~\cite{Ron:2010ua,Rubinfeld:2011ik} and a book~\cite{Goldreich:2011cg}.

In this paper, we consider testing properties of functions $f:\Fp^n \to \bit$, where $p$ is a fixed prime.
We say that a function $f$ is \emph{$\epsilon$-far} from a property $\calP$ if we must modify an $\epsilon$-fraction of values of $f$ to make it satisfy $\calP$.
In other words, for any function $g$ that satisfies $\calP$, we have $\Pr_{x \in \Fp^n}[f(x) \neq g(x)] \geq \epsilon$, where $x$ is chosen uniformly at random.
Otherwise, the function $f$ is called \emph{$\epsilon$-close} to $\calP$.
We formally define testers as follows.
\begin{definition}[Tester]
  An algorithm is called an $\epsilon$-tester for a property $\calP$ if,
  given a query access to a function $f:\Fp^n \to \bit$,
  with probability at least $2/3$,
  it accepts when $f$ satisfies $\calP$ and rejects when $f$ is $\epsilon$-far from $\calP$.
\end{definition}
The parameter $\epsilon$ is called the \emph{proximity parameter}.
The probability threshold $2/3$ is not so important since we can make it $1-\delta$ for any $\delta >0$ by running the tester $O(\log 1/\delta)$ times and take the majority of outputs.
If a property is testable with query complexity that depends only on $\epsilon$ (and $\calP$) but not on $n$, 
it is called \emph{locally testable} or \emph{strongly testable}.
A tester is called a \emph{one-sided error tester} for a property $\calP$ if it always accepts functions satisfying $\calP$, and is called a \emph{two-sided error tester} otherwise.

In this paper, we consider the (two-sided error) testability of affine-invariant properties.
For a matrix $L \in \Fp^{m \times n}$ and a vector $c \in \Fp^m$,
the pair $A=(L,c)$ is called an \emph{affine transformation}, and it maps $x \in \Fp^n$ to $Lx + c \in \Fp^m$.
We say that an affine transformation $A=(L,c)$ is \emph{non-singular} if $L$ is non-singular.
A property $\calP$ of functions is called \emph{affine-invariant} if,
for any function $f:\Fp^n \to \bit$ satisfying $\calP$ and any non-singular affine transformation $A:\Fp^n \to \Fp^n$, the function $f \circ A:\Fp^n \to \bit$ also satisfies $\calP$.
Many affine-invariant properties are known to be locally testable, including linearity~\cite{Blum:1993cn}, the property of being a low-degree polynomial~\cite{Alon:2005jl}, and Fourier sparsity~\cite{Gopalan:2011jj}.
Kaufman and Sudan~\cite{Kaufman:2008jm} made explicit that these properties are affine-invariant and initiated a general study of the testability of affine-invariant properties.
In particular, they asked for necessary and sufficient conditions of local testability of affine-invariant properties.
The main contribution of this paper is answering their question by giving a characterization of locally testable affine-invariant properties.

Alon et al.~\cite{Alon:2009gn} showed a combinatorial characterization of locally testable properties for (dense) graphs.
The characterization is based on Szemer{\'e}di's regularity lemma~\cite{Szemeredi:1975tr}, 
which roughly claims that any graph can be partitioned into a constant number of parts so that every pair of parts forms a random bipartite graph.
Their characterization indicates that a graph property is locally testable if and only if densities of these bipartite graphs determine whether the property holds.
A point here is that, if a graph property is locally testable, then whether it holds only depends on a constant-size sketch of the input graph, namely the set of densities.

When studying affine-invariant properties,
higher-order Fourier analysis provides us a way to extract such a constant-size sketch from a function.
The main technical tools we exploit here are the decomposition theorems shown in~\cite{Bhattacharyya:2013ii},
which roughly claim that any function $f:\Fp^n \to \bit$ can be decomposed as $f = f' + f''$, where $f'$ is a ``structured'' part of $f$ and $f'$ is a ``pseudo-random'' part of $f$.
Here $f'$ is structured in the sense that it can be expressed as $f' = \Gamma(P_1,\ldots,P_C)$ for some function $\Gamma$ and non-classical polynomials $P_1,\ldots,P_C$ of constant degrees.
The precise definition of a non-classical polynomial is given later (Section~\ref{sec:preliminaries}).
Here we only have to understand that, besides degree, 
a non-classical polynomial $P$ has a parameter called \emph{depth}, which is less than the degree of $P$.
In this paper, if we refer to a polynomial, it is always a non-classical polynomial.
We can assume that the range of a non-classical polynomial $P$ of depth $h$ is $\bbU_{h+1} = \frac{1}{p^{h+1}}\bbZ/\bbZ$, the set of multiples of $\frac{1}{p^{h+1}}$ in $[0,1]$.
Hence, $\Gamma$ is a function from $\prod_{i=1}^C \bbU_{h_i+1} $ to $[0,1]$, where $h_i$ is the depth of the polynomial $P_i$ for each $i \in \{1,\ldots,C\}$.

In our setting, $\Gamma(P_1,\ldots,P_C)$ will be used as a sketch of a function $f$.
An issue here is that the polynomials $P_1,\ldots,P_C$ depend on $n$ values,
and thus they may not have constant-size representations.
However, we can ensure in the decomposition that the polynomial sequence $\bfP = (P_1,\ldots,P_C)$ has a \emph{high rank}.
We give a precise definition of rank later (Section~\ref{sec:preliminaries}).
What we need to know here is that,
if a polynomial sequence $\bfP = (P_1,\ldots,P_C)$ has a high rank and we sample $x \in \Fp^n$ uniformly at random,
then the distribution of the tuple $\bfP(x) = (P_1(x),\ldots,P_C(x))$ looks almost random in $\prod_{i=1}^C \bbU_{h_i+1}$.
Hence, provided that the rank is high,
many properties of $f' = \Gamma(P_1,\ldots,P_C)$ are determined only by the function $\Gamma$, degrees of $P_1,\ldots,P_C$, and depths of $P_1,\ldots,P_C$.

The function $f''$ is pseudo-random in the sense that its Gowers norm of order $d$, denoted $\|f''\|_{U^d}$, is small,
where $d$ is more than the maximum degree of $P_1,\ldots,P_C$.
The \emph{Gowers norm} of order $d$ measures correlation with polynomials of degree less than $d$ (See Section~\ref{sec:preliminaries} for further details).
We can show that, if $\|f''\|_{U^d}$ is small, 
then it does not significantly affect the distribution of $f$ restricted to a random affine subspace of a constant dimension in $\Fp^n$.
This distribution is very important since it is known that,
if an affine-invariant property $\calP$ is locally testable,
then there is a \emph{canonical tester} for $\calP$ with a constant query complexity, whose answer only depends on the distribution~\cite{Bhattacharyya:2010gb}.
Hence, when studying the testability of affine-invariant properties,
it turns out that we only have to look at $\Gamma$, degrees, and depths.

To explicitly express the form of a structured part, we define \emph{regularity-instances} as follows.
Here, $\bbN$ denotes the set of non-negative integers.
\begin{definition}[Regularity-instance]\label{def:regularity-instance}
  A regularity-instance $I$ is a tuple of
  \begin{itemize}
  \item an error parameter $\gamma \in \bbR$ with $\gamma > 0$,
  \item a structure function $\Gamma:\prod_{i=1}^C \bbU_{h_{i+1}} \to [0,1]$,
  \item a complexity parameter $C \in \bbN$,
  \item a degree-bound parameter $d \in \bbN$,
  \item a degree parameter $\bfd = (d_1,\ldots,d_C) \in \bbN^C$ with $d_i < d$ for each $i \in \{1,\ldots,C\}$,
  \item a depth parameter $\bfh = (h_1,\ldots,h_C) \in \bbN^C$ with $h_i < d_i$ for each $i \in \{1,\ldots,C\}$, and
  \item a rank parameter $r \in \bbN$.
  \end{itemize}
  The \emph{complexity} of the regularity-instance is $\max(1/\gamma,C,d,r)$.
\end{definition}
Here, the name ``regularity-instance'' is taken from~\cite{Alon:2009gn}.
We define the property of satisfying a regularity-instance as follows.
\begin{definition}[Satisfying a regularity-instance]\label{def:satisfy-regularity-instance}
  A function $f:\Fp^n \to \bit$ is said to satisfy a regularity-instance $I = (\gamma,\Gamma,C,d,\bfd,\bfh,r)$ if 
  there is a function $\Upsilon:\Fp^n \to [-1,1]$ and a polynomial sequence $\bfP = (P_1,\ldots,P_C)$ on $n$ variables such that
  \begin{itemize}
  \item $f(x) = \Gamma(\bfP(x)) + \Upsilon(x)$ for any $x \in \Fp^n$,
  \item $P_i$ has degree exactly $d_i$ and depth exactly $h_i$ for each $i \in \{1,\ldots,C\}$,
  \item the rank of the polynomial sequence $\bfP$ is at least $r$, and
  \item $\|\Upsilon\|_{U^d} \leq \gamma$.
  \end{itemize}
\end{definition}


The first requirement we need in order to obtain our characterization of locally testable properties is that the property of satisfying a regularity-instance is locally testable, provided that the rank parameter is chosen to be sufficiently high.
\begin{theorem}\label{the:regularity-instance->testable}
  For any $\epsilon > 0$ and any regularity-instance $I = (\gamma,\Gamma,C,d,\bfd,\bfh,r)$ with $r \geq r_{\ref{the:regularity-instance->testable}}(\gamma,\epsilon,C,d)$,
  there is an $\epsilon$-tester for the property of satisfying $I$ with a constant query complexity.
\end{theorem}
What we must be careful about here is that,
in order to satisfy $I = (\gamma,\Gamma,C,d,\bfd,\bfh,r)$, 
the input function $f(x)$ should be close to $\Gamma(\bfP(x))$ for a polynomial sequence $\bfP = (P_1,\ldots,P_C)$ such that the polynomial $P_i$ has degree \emph{exactly} $d_i$ and depth \emph{exactly} $h_i$ for each $i \in \{1,\ldots,C\}$.
These conditions are important when studying locally testable properties since the distribution of a function restricted to a random affine subspace is determined by exact degrees and depths, but not by their upper bounds.
To ensure that these conditions are satisfied, we need the rank condition in Theorem~\ref{the:regularity-instance->testable}.
We note that the property of satisfying a regularity-instance $I$ is affine-invariant, that is, closed under non-singular affine transformations, but is \emph{not} closed under all affine transformations since the degree and the depth of a polynomial and the rank of a polynomial sequence may decrease through affine transformations.
This means that we can only test the property with two-sided error since we can only look at the restriction of the input function to an affine subspace.

Suppose that we replace the condition ``exactly'' by ``at most'' and drop the rank condition in Definition~\ref{def:satisfy-regularity-instance}.
Under this definition, the property of satisfying a regularity-instance is closed under all affine transformations.
Indeed, 
if we further require that the function $\Upsilon$ is constantly zero,
this property is called a \emph{degree-structural property},\footnote{In~\cite{Bhattacharyya:2013ii}, degree-structural properties are defined using a constant number of regularity-instances.}
and known to be locally testable with one-sided error~\cite{Bhattacharyya:2013ii}.

One might be skeptical about the usefulness of Theorem~\ref{the:regularity-instance->testable} since it is unclear whether there is indeed a polynomial sequence $\bfP$ that has the required rank as it depends on the size of $\bfP$.
To clarify this problem, we recall the polynomial regularity lemma~\cite{Tao:2011dw,Bhattacharyya:2013ii},
which claims that, given any sequence of $C$ polynomials with degrees at most $d$ and a function $r:\bbN \to \bbN$,
we can ``refine'' the sequence to make a new sequence of $C'$ polynomials of degrees at most $d$ with rank at least $r(C')$ for some constant $C'$ that depends on $C$, $d$, and $r$.
We also note that we cannot remove the dependency to $\gamma$ and $\epsilon$ from the rank condition in Theorem~\ref{the:regularity-instance->testable} since,
the smaller they are,
the more we want the polynomial sequence to behave randomly in order to achieve a local tester.

For a parameter $\delta > 0$, we say that a regularity-instance with complexity $C$ and degree parameter $d$ has a \emph{high rank with respect to $\delta$} if its rank parameter is at least $r_{\ref{the:regularity-instance->testable}}(\gamma,\delta/8,C,d)$.
The reason we use $\delta/8$ instead of $\delta$ is technical and will be discussed in Section~\ref{sec:reducible->testable}.
The following definition aims to capture function properties that are locally testable via testing a certain set of regularity-instances.
\begin{definition}[Regular-reducible]\label{def:regular-reducible}
  A property $\calP$ is regular-reducible if,
  for any $\delta>0$,
  there exists $s$ such that,
  for any $n \in \bbN$, there is a family $\calI$ of at most $s$ regularity-instances each with a complexity at most $s$ and high rank with respect to $\delta$ with the following properties.
  For every $\epsilon > \delta$ and a function $f:\Fp^n \to \bit$,
  \begin{itemize}
  \item if $f$ satisfies $\calP$, then for some $I \in \calI$, $f$ is $\delta$-close to satisfying $I$, and
  \item if $f$ is $\epsilon$-far from satisfying $\calP$, then for any $I \in \calI$, 
    $f$ is $(\epsilon-\delta)$-far from satisfying $I$.
  \end{itemize}
\end{definition}
Now we are ready to state our characterization of locally testable affine-invariant properties.
\begin{theorem}\label{the:testable->regular-reducible}
  If an affine-invariant property is locally testable, then it is regular-reducible.
\end{theorem}
\begin{theorem}\label{the:regular-reducible->testable}
  If an affine-invariant property $\calP$ is regular-reducible,
  then it is locally testable.
\end{theorem}
These theorems give a complete answer to the main question in the study of the (two-sided error) testability of affine-invariant properties.
On the other hand, they are interesting only qualitatively since the query complexity of the tester given by Theorem~\ref{the:regular-reducible->testable} is rather horrible -- Ackermann-like function that depends on $1/\epsilon$.
We note though that recent works by Kalyanasundaram and Shapira~\cite{Kalyanasundaram:2013gv} and by Conlon and Fox~\cite{Conlon:2012es} suggest that the very rapid growth of the query complexity function is in fact inherent in the nature of the problem.

As is evident from Definition~\ref{def:regular-reducible},
our characterization is not a quick recipe for inferring whether a given property is locally testable.
In particular, the rank condition may be an obstacle to obtain a set of regularity-instances to which a property is regular-reducible.
Hence, we provide a variant of regular-reducibility that is more accessible.
A \emph{rank-oblivious regularity-instance} is a regularity-instance without the rank parameter.
We define the property of satisfying rank-oblivious regularity-instances as follows.
\begin{definition}
  A function $f:\Fp^n \to \bit$ is said to satisfy a rank-oblivious regularity-instance $I = (\gamma,\Gamma,C,d,\bfd,\bfh)$ if there is a function $\Upsilon: \Fp^n \to [-1,1]$ and a polynomial sequence $\bfP = (P_1,\ldots,P_C)$ such that
  \begin{itemize}
  \item $f(x) = \Gamma(\bfP(x)) + \Upsilon(x)$,
  \item $P_i$ has degree $d_i$ and depth $h_i$ for each $i \in \{1,\ldots,C\}$.
  \item $\|\Upsilon\|_{U^d} \leq \gamma$.
  \end{itemize}
\end{definition}

Now we define a variant of regular-reducibility using rank-oblivious regularity-instances.
\begin{definition}[rank-obliviously regular-reducible]\label{def:rank-obliviously-regular-reducible}
  A property $\calP$ is rank-obliviously regular-reducible if,
  for any $\delta > 0$,
  there exists $s$ such that,
  for any $n \in \bbN$,
  there is a family $\calI$ of at most $s$ rank-oblivious regularity-instances each of complexity at most $s$ with the following properties.
  For every $\epsilon > \delta$ and a function $f:\Fp^n \to \bit$,
  \begin{itemize}
  \item if $f$ satisfies $\calP$, then for some $I \in \calI$,
  $f$ is $\delta$-close to satisfying $\calI$, and
  \item if $f$ is $\epsilon$-far from satisfying $\calP$,
  then for any $I \in \calI$,
  $f$ is $(\epsilon-\delta)$-far from satisfying $\calI$.
  \end{itemize}
\end{definition}

From the following theorem, we can also show local testability using rank-oblivious regular-reducibility.
\begin{theorem}\label{the:rank-obliviously-reducible->reducible}
  If a property is rank-obliviously regular-reducible,
  then the property is regular-reducible.
  In particular, the property is locally testable.
\end{theorem}
Unfortunately, it seems that the converse does not hold in general.
Nonetheless, Theorem~\ref{the:rank-obliviously-reducible->reducible} is useful to show local testability of interesting properties such as 
degree-structural properties~\cite{Bhattacharyya:2013ii}.
As an illustrative example,
using Theorem~\ref{the:rank-obliviously-reducible->reducible},
we show that the property of being a (classical) low-degree polynomial is locally testable in Section~\ref{sec:rank-oblivious-reducibility}.
We note that,
for this particular property,
an almost tight result is already known~\cite{Alon:2005jl}.

\subsection{Related work}
This work is a part of a sequence of works investigating the relationship between affine-invariance and testability of properties. 
As described, Kaufman and Sudan~\cite{Kaufman:2008jm} initiated the program.
There have been a number of studies on one-sided error testability of affine-invariant properties~\cite{Bhattacharyya:2011bf,Kral:2012va,Shapira:2009eo,Bhattacharyya:2010gb,Bhattacharyya:2012ud,Bhattacharyya:2013ii}.
In particular, Bhattacharyya~et~al.~\cite{Bhattacharyya:2010gb} conjectured that every subspace hereditary property is locally testable with one-sided error,
where a property $\calP$ is \emph{subspace hereditary} if,
for any function $f:\Fp^n \to \bit$ satisfying $\calP$,
its restriction to any affine subspace of $\Fp^n$ also satisfies $\calP$.
Resolving this conjecture would yield a combinatorial characterization of affine-invariant properties that are locally testable with one-sided error.
Although the conjecture has not yet been confirmed or refuted, 
Bhattacharyya~et~al.~\cite{Bhattacharyya:2013ii} showed that any subspace hereditary property of ``bounded complexity'' is locally testable with one-sided error.
The precise definition of complexity is technical and we omit here;
however, it is an integer associated with a property,
and all natural affine-invariant properties that we know of have bounded complexity.
Recently, Hatami and Lovett~\cite{Hatami:2013ux} showed that,
if an affine property $\calP$ is locally testable,
then the distance to $\calP$ can be estimated with a constant number of queries.
The main technical tool used to achieve these general results is higher-order Fourier analysis and especially decomposition theorems developed in~\cite{Bhattacharyya:2012ud,Bhattacharyya:2013ii}.
Among many works on higher-order Fourier analysis, we refer the reader to a book~\cite{Tao:2012ue} for an overview of the contemporary theory related to this topic.

These studies of the testability of affine-invariant properties parallel work in testability of graph properties.
In the \emph{adjacency graph model}~\cite{Goldreich:1998wa},
a graph $G=(V,E)$ is given as a query access to its adjacency matrix.
That is, if we specify two vertices, then the oracle returns whether there is an edge between them in $G$.
We say that a graph $G$ is \emph{$\epsilon$-far} from a property if we must add or remove at least $\epsilon |V|^2$ edges to make $G$ satisfy the property.
In this model, we can also locally test many properties such as $3$-colorability~\cite{Goldreich:1998wa} and triangle-freeness~\cite{Alon:2000dx}.
Alon and Shapira~\cite{Alon:2008ck} showed that a (natural) graph property is locally testable with one-sided error if and only if the property is hereditary,
where a graph property $\calP$ is \emph{hereditary} if,
for any graph satisfying $\calP$, its any induced subgraph also satisfies $\calP$.
Fischer and Newman~\cite{Fischer:2007bq} showed that if a graph property is locally testable, then the distance to the property can be estimated with a constant number of queries.
Based on this result, Alon~et~al.~\cite{Alon:2009gn} finally obtained a combinatorial characterization of locally testable properties.
Our work can be seen as an analogue of~\cite{Alon:2009gn} for affine-invariant properties.
Similarly to~\cite{Alon:2009gn}, our proof also uses the result of the estimation of distances to affine-invariant properties~\cite{Hatami:2013ux}.

Finally, we mention that a characterization of locally testable properties is known in a very different setting.
In the \emph{assignment testing of constraint satisfaction problems (CSPs)}, 
we are given an instance of a CSP and a query access to an assignment for the instance,
and we want to test whether the assignment is a satisfying assignment or far from being so.
Depending on the constraints we are allowed to use,
CSPs can express many different problems and the query complexity to test drastically changes from constant to linear (in the number of variables).
Recently, Bhattacharyya and Yoshida~\cite{Bhattacharyya:2013fa} completely classified Boolean constraints in terms of query complexity.

\subsection{Proof sketch}
We now give proof sketches of our main theorems.

We start by discussing Theorem~\ref{the:regularity-instance->testable}.
Fix a proximity parameter $\epsilon > 0$ and a regularity-instance $I = (\gamma,\Gamma,C,d,\bfd,\bfh,r)$ satisfying the rank condition.
Our tester is very simple.
That is,
for $\delta = \delta(\gamma,\epsilon,C,d)$ and $m = m(\gamma,\epsilon,C,d)$,
we choose a random affine embedding $A:\Fp^m \to \Fp^n$,
and accept if $f \circ A$ is $\delta$-close to satisfying $I$ and reject otherwise.
Here, an \emph{affine embedding} is an injective affine transformation.

Suppose that $f$ satisfies the regularity-instance $I$.
That is, $f(x) = \Gamma(\bfP(x)) + \Upsilon(x)$ for a polynomial sequence $\bfP$ on $n$ variables with degree $\bfd$, depth $\bfh$, and rank at least $r$, and a function $\Upsilon:\Fp^n \to [-1,1]$ with $\|\Upsilon\|_{U^d} \leq \gamma$.
Then, $f \circ A$ can be written as $f(Ax) = \Gamma((\bfP \circ A)(x)) + \Upsilon(Ax)$.
It is not difficult to show that,
with high probability over the choice of $A$,
$\bfP \circ A$ has the same degree, depth, and rank as $\bfP$,
and $\|\Upsilon \circ A\|_{U^d}$ is only slightly larger than $\gamma$.
We can then show that,
by perturbing $f \circ A$ up to a $\delta$-fraction,
we can decrease the Gowers norm of $\Upsilon \circ A$ to $\gamma$.
Hence, $f \circ A$ is $\delta$-close to satisfying $I$.

Now suppose that $f$ is $\epsilon$-far from satisfying $I$.
Assume that (with high probability) $f \circ A$ is $\delta$-close to satisfying $I$.
In such a case, $f \circ A$ can be written as 
\[
  f(Ax) = \Gamma(\bfP'(x)) + \Upsilon'(x) + \Delta'(x)
\] 
for a polynomial sequence $\bfP' = (P'_1,\ldots,P'_C)$ on $m$ variables with degree $\bfd$, depth $\bfh$, and rank at least $r$, a function $\Upsilon':\Fp^m \to [-1,1]$ with $\|\Upsilon'\|_{U^d} \leq \gamma$, and a function $\Delta':\Fp^m \to \{-1,0,1\}$ with $\|\Delta\|_1 \leq \delta$.
Our strategy is to construct a polynomial sequence $\bfP$ on $n$ variables from $\bfP'$ with degree $\bfd$, depth $\bfh$, and rank at least $r$ so that $\|f - \Gamma \circ \bfP\|_{U^d}$ is slightly larger than $\gamma$. 
Then, by slightly perturbing $f$ (up to an $\epsilon$-fraction),
we can decrease its Gowers norm to $\gamma$.
Hence, $f$ is $\epsilon$-close to satisfying $I$ and we reach a contradiction.

Using the regularity lemma, we can decompose the input function $f:\Fp^n \to \bit$ as $f = f_1+f_2+f_3$.
Here, $f_1(x) = \Sigma(\bfR(x))$ for some function $\Sigma:\bbT^{|\bfR|} \to [0,1]$ and a high-rank polynomial sequence $\bfR$, where $\bbT = \bbR / \bbZ$ is the circle group.
Also, $f_2:\Fp^n \to [-1,1]$ has small $L_2$-norm\footnote{We did not mention $f_2$ when discussing the polynomial regularity lemma earlier in this introduction. However it turns out that the Gowers norm of $f_2$ can be made small. Thus $f_2$ can also be seen as a pseudo-random part of $f$.}, and $f_3:\Fp^n \to [-1,1]$ has a small Gowers norm.
With this decomposition,
by letting $\bfR' = \bfR \circ A$,
we can express $f \circ A$ as 
\[
  f(Ax) = \Sigma(\bfR'(x)) + f_2(Ax) + f_3(Ax).
\]
Hence, we have obtained two ways of expressing the function $f \circ A$.

Now we introduce the notion of a factor.
Note that a polynomial sequence $\bfQ = (Q_1,\ldots,Q_D)$ on $m$ variables of depth $(h_1,\ldots,h_D)$ defines a partition of the space $\prod_{i=1}^D \bbU_{h_i+1}$.
That is, for any tuple $(b_1,\ldots,b_D)$ with $b_i \in \bbU_{h_i+1}$ for each $i \in \{1,\ldots,D\}$,
there is a corresponding part, called an \emph{atom}, $\{x \in \Fp^m \mid (Q_1(x),\ldots,Q_D(x) = (b_1,\ldots,b_D)\}$.
We call the partition the \emph{factor} defined by $\bfQ$ and denote it by $\calB(\bfQ)$.

Now we come back to our argument on Theorem~\ref{the:regularity-instance->testable}.
Using a variant of the polynomial regularity lemma, 
given two polynomial sequences $\bfR'$ and $\bfP'$,
we can find a polynomial sequence $\bfS'$ of degree less than $d$ with the following property:
By letting $\overline{\bfR'} = \bfR' \cup \bfS'$,
the factor $\calB(\overline{\bfR'})$ is a refinement of both the factor $\calB(\bfR')$ and the factor $\calB(\bfP')$.
Hence, for each $i \in \{1,\ldots,C\}$,
we can find a function $\Gamma_i:\bbT^{|\overline{\bfR'}|} \to \bbT$ such that 
$P'_i(x) = \Gamma_i(\overline{\bfR'}(x))$.
Then we can write 
\[
  f(Ax) = \Sigma(\bfR'(x)) + f_2(Ax) + f_3(Ax) = \Gamma(\Gamma_1(\overline{\bfR'}(x)),\ldots,\Gamma_C(\overline{\bfR'}(x))) + \Upsilon'(x) + \Delta'(x).
\]
Let $\Phi(x) = f_2(Ax) + f_3(Ax) - \Upsilon'(x) - \Delta'(x)$.
Since the factor $\calB(\overline{\bfR'})$ is a refinement of the factor $\calB(\bfR')$,
$\Phi(x)$ is a function that is constant on each atom of $\calB(\overline{\bfR'})$.
Hence, we can express
\[
  \Sigma(\bfR'(x)) + \Phi(\overline{\bfR'}(x)) = \Gamma(\Gamma_1(\overline{\bfR'}(x)),\ldots,\Gamma_C(\overline{\bfR'}(x))).
\]
Here, we reuse the symbol $\Phi$.

Using the condition that $\overline{\bfR'}$ has a high rank,
we can show that, for every $b \in \bbT^{|\overline{\bfR'}|}$ in the range of $\overline{\bfR'}$ (determined by the depth of $\overline{\bfR'}$),
by choosing $x \in \Fp^m$ uniformly at random,
there is a positive probability that $\overline{\bfR'}(x)$ takes the value $b$.
Hence, for every $a \in \bbT^{|\bfR'|}$ in the range of $\bfR'$ and every $b \in \bbT^{|\bfS'|}$ in the range of $\bfS'$,
we have
\[
  \Sigma(a) + \Phi(a,b) = \Gamma(\Gamma_1(a,b),\ldots,\Gamma_C(a,b)).
\]

Now we define $\overline{\bfR} = \bfR \cup (\bfS' \circ A^+)$, where $A^+:\Fp^n \to \Fp^m$ is an affine transformation with $A^+ A = I_m$.
Further we define $P_i = \Gamma_i \circ \overline{\bfR}$ for each $i \in \{1,\ldots,C\}$ and $\bfP = (P_1,\ldots,P_C)$.
From the observation above,
we have for any $x \in \Fp^n$,
\[
  \Sigma(\bfR(x)) + \Phi(\overline{\bfR}(x)) = \Gamma(\bfP(x)).
\]
Since $\bfP'$ can be obtained from $\bfP$ by applying an affine transformation,
for each $i \in \{1,\ldots,C\}$,
the degree and the depth of $P_i$ is at least those of $P'_i$,
and the rank of $\bfP$ is at least that of $\bfP'$.
Using the fact that $\bfP$ is of high rank,
we can indeed show that,
for each $i \in \{1,\ldots,C\}$,
the degree and the depth of $P_i$ are exactly the same as those of $P'_i$.
Hence, $\Gamma \circ \bfP$ satisfies conditions required by the regularity-instance $I$.

Recalling that  $f = f_1+f_2+f_3$ with $f_1 = \Sigma \circ \bfR$,
we have for any $x \in \Fp^n$,
\[
  f(x) - \Gamma(\bfP(x)) = f_2(x) + f_3(x) - \Phi(\overline{\bfR}(x)).
\]
Then we want to show that $\|f_2 + f_3 - \Phi(\overline{\bfR})\|_{U^d} \leq \gamma + o(\gamma)$.
It is clear that the Gowers norms of $f_2$ and $f_3$ are $o(\gamma)$ from the property of the decomposition.
We show that $\|\Phi \circ \overline{\bfR}\|_{U^d} \leq \gamma + o(\gamma)$ by showing that $\|\Phi \circ \overline{\bfR'}\|_{U^d} \leq \gamma + o(\gamma)$ and $|\|\Phi \circ \overline{\bfR}\|_{U^d} - \|\Phi \circ \overline{\bfR'}\|_{U^d}| = o(\gamma)$.
To show the former, recall that $\Phi \circ \overline{\bfR'} = f_2 \circ A + f_3 \circ A - \Upsilon' - \Delta'$.
It is not difficult to show that, with high probability over the choice of $A$, $\|f_2 \circ A\|_{U^d}$ and $\|f_3 \circ A\|_{U^d}$ are small as $\|f_2\|_{U^d}$ and $\|f_3\|_{U^d}$ are small.
We have $\|\Upsilon'\|_{U^d} \leq \gamma$ from the assumption.
Since $\|\Delta'\|_1 \leq \delta$,
by choosing $\delta$ small enough and using the relation between the $L_1$ norm and the Gowers norm, we can also bound $\|\Delta'\|_{U^d}$.
Showing the latter is technical,
but basically it holds since the restrictions of $\overline{\bfR}$ and $\overline{\bfR'}$ to a random affine subspace of dimension $d$ look similar because of their high ranks and the Gowers norm only depends on these restrictions.

Now that the Gowers norm of $f(x) - \Gamma(\bfP(x))$ is at most $\gamma + o(\gamma)$,
we can show that,
by perturbing $f$ up to an $\epsilon$-fraction, we can obtain a function that satisfies the regularity-instance $I$, which contradicts the assumption that $f$ is $\epsilon$-far from satisfying $I$.

The proof of Theorem~\ref{the:testable->regular-reducible} is similar to~\cite{Alon:2009gn}.
From the whole high-rank regularity-instances of complexity bounded by some constant,  
we take a ``$\tau$-net'' $\calR$ for suitably chosen $\tau$.
Hence, any function $f:\Fp^n \to \bit$ is $\tau$-close to some instance $I \in \calR$ in a certain sense.
From the argument of canonical testers,
if a property $\calP$ is locally testable,
then whether or not $f$ satisfies $\calP$ should depend only on the distribution of the restriction of $f$ to a random affine subspace.
Since $f$ is $\tau$-close to a regularity-instance $I$,
we can approximate the distribution using the structure function of $I$.
Hence, as $\calI$ in Definition~\ref{def:regular-reducible},
we choose regularity-instances $I$ for which the canonical tester accepts under the distribution associated with $I$.

The proof of Theorem~\ref{the:regular-reducible->testable} is almost immediate once we have Theorem~\ref{the:regularity-instance->testable}.
However, note that we want to test whether the input function is \emph{close to} satisfying some regularity-instance in $\calI$.
Hence, we use a recent result by Hatami and Lovett~\cite{Hatami:2013ux},
which states that,
if an affine-invariant property is locally testable,
then we can estimate the distance to the property with a constant number of queries.
One issue with which we must be careful is that Theorem~\ref{the:regularity-instance->testable} does not claim that a fixed regularity-instance is locally testable since the rank condition depends on $\epsilon$.
Nonetheless we can apply the result by Hatami and Lovett since,
when we want to distinguish the case that a function is $\epsilon_1$-close to a property $\calP$ from the case that it is $\epsilon_2$-far from $\calP$ for $0 < \epsilon_1 < \epsilon_2 < 1$,
we only require the testability of $\calP$ with the proximity parameter $O(\epsilon_2 - \epsilon_1)$.
This condition indeed holds since we have included the high-rank condition in Definition~\ref{def:regular-reducible}.

\subsection{Discussion}
We have obtained a characterization of locally testable affine-invariant properties using decomposition theorems.
A natural remaining problem in testability of affine-invariant properties is obtaining a characterization of properties that are locally testable with one-sided error.
As we have mentioned, Bhattacharyya~et~al.~\cite{Bhattacharyya:2010gb} conjectured that such properties are essentially subspace hereditary properties.
Another natural open problem is to find characterizations or at least sufficient conditions of properties that are testable with query complexity that is polynomial in $1/\epsilon$.

Permutation-invariant properties have also been well studied in the literature,
where a function property $\calP$ is called \emph{permutation-invariant} if,
for any function $f:\bit^n \to \bit$ satisfying $\calP$,
any function obtained from $f$ by relabeling input bits also satisfies $\calP$.
It seems that a permutation-invariant property $\calP$ is locally testable whenever any function satisfying $\calP$ has a ``concise'' representation~\cite{Blais:2009kt,Diakonikolas:2007gm,Yoshida:2012dm}.
However, how to formalize the idea to obtain a full characterization of locally testable properties in this setting remains unclear.

\subsection{Organization}
In Section~\ref{sec:preliminaries},
we review definitions and basic results in higher-order Fourier analysis.
In Section~\ref{sec:small-perturbation},
we show that,
if a function $f$ has a structure part and a pseudo-random part similar to a regularity-instance $I$,
then $f$ is indeed close to satisfying $I$.
In Section~\ref{sec:regularity-instance->testable},
we show that any regularity-instance is locally testable, provided that the rank parameter is sufficiently high (Theorem~\ref{the:regularity-instance->testable}).
In Section~\ref{sec:testable->regular-reducible} we show that any locally testable property is regular-reducible (Theorem~\ref{the:testable->regular-reducible}) and in Section~\ref{sec:reducible->testable} we show that any regular-reducible property is locally testable (Theorem~\ref{the:regular-reducible->testable}).
In Section~\ref{sec:rank-oblivious-reducibility}, we show that rank-oblivious regular-reducibility implies regular-reducibility (Theorem~\ref{the:rank-obliviously-reducible->reducible}) and provide one of its applications.

\section{Preliminaries}\label{sec:preliminaries}
For $f : \Fp^n \to \bbC$ we denote $\|f\|_1= \E_x[|f(x)|]$, $\|f\|_2= \E_x[|f(x)|^2]$ where $x \in \Fp^n$ is chosen uniformly at random and $\|f\|_\infty = \max_x|f(x)|$.
Note that $\|f\|_1 \leq  \|f\|_2 \leq \|f\|_\infty$.
The expression $o_m(1)$ denotes quantities which approach zero as $m$ grows.
We shorthand $x \pm \epsilon$ for any quantity in $[x - \epsilon, x+\epsilon]$.
For probability distributions $\mu$ and $\mu'$ over the domain $A$,
we define the \emph{statistical distance} $\dtv(\mu,\mu')$ between $\mu$ and $\mu'$ by 
$\dtv(\mu,\mu') = \frac{1}{2}\sum_{a \in D}|\mu(a) - \mu(a')| = \frac{1}{2}\|\mu - \mu'\|_1$.
In this paper, bold symbols indicate sets (of integers, polynomials, etc).

In what follows, we introduce definitions and results about higher-order Fourier analysis.
Most of the material in this section is directly quoted from~\cite{Bhattacharyya:2013ii,Hatami:2013ux}.

\subsection{Uniformity norms and non-classical polynomials}

\begin{definition}[Multiplicative Derivative]
  Given a function $f:\Fp^n \to \bbC$, and an element $h \in \Fp^n$,
  define the multiplicative derivative in direction $h$ of $f$ to be the function $\Delta_hf:\Fp^n \to \bbC$ satisfying $\Delta_h f(x) = f(x+h)\overline{f(x)}$ for all $x \in \Fp^n$.
\end{definition}

\begin{definition}[Gowers norm]
  Given a function $f:\Fp^n \to \bbC$ and an integer $d \geq 1$, the Gowers norm of order $d$ for $f$ is given by
  \begin{align*}
    \|f\|_{U^d} = \left|\E_{x,y_1,\ldots,y_d \in \Fp^n}[(\Delta_{y_1}\Delta_{y_2} \cdots \Delta_{y_d}f)(x)] \right|^{1/2^d}.
  \end{align*}
\end{definition}

Note that,
as $\|f\|_{U^1} = |\E[f]|$, the Gowers norm of order $1$ is only a semi-norm.
However for $d > 1$, it is not difficult to show that $\|\cdot\|_{U^d}$ is indeed a norm.

The following lemma connects the Gowers norm and the $L_1$ norm.
\begin{lemma}[Claim 2.21 of~\cite{Hatami:2013ux}]\label{lem:bound-gowers-norm-by-l1-norm}
  Let $f:\Fp^n \to [-1,1]$. For any $d \in \bbN$,
  \[
    \|f\|_{U^d} \leq \|f\|_1^{1/2^d}.
  \]
\end{lemma}

If $f = e^{2\pi i P/p}$ where $P: \Fp^n \to \Fp$ is a polynomial of degree less than $d$,
then $\|f\|_{U^d} = 1$.
If $d < p$ and $\|f\|_{\infty} \leq 1$,
then in fact,
the converse holds,
meaning that any function $f: \Fp^n \to \bbC$ satisfying $\|f\|_\infty \leq 1$ and $\|f\|_{U^d} = 1$  is of this form.
But when $d \geq p$,
the converse is no longer true.
In order to characterize functions $f:\Fp^n \to \bbC$ with $\|f\|_\infty \leq 1$ and $\|f\|_{U^d} = 1$, we define the notion of non-classical polynomials.

Non-classical polynomials might not be necessarily $\Fp$-valued.
We need to introduce some notation.
Let $\bbT$ denote the circle group $\bbR / \bbZ$.
This is an abelian group with group operation denoted $+$.
For an integer $k \geq 0$,
let $\bbU_k$ denote $\frac{1}{p^k}\bbZ / \bbZ$, a subgroup of $\bbT$.
Let $\iota : \Fp \to \bbU_1$ be the injection $x \mapsto \frac{|x|}{p} \bmod 1$,
where $|x|$ is the standard map from $\Fp$ to $\{0,1,\ldots,p-1\}$.
Let $\sfe:\bbT \to \bbC$ denote the character $\sfe(x) = e^{2\pi i x}$.

\begin{definition}[Additive Derivative]
  Given a function $P : \Fp^n \to \bbT$ and an element $h \in \Fp^n$,
  define the additive derivative in direction $h$ of $f$ to be the function $D_hP : \Fp^n \to \bbT$ satisfying $D_hP(x)=P(x+h)-P(x)$ for all $x \in \Fp^n$.
\end{definition}
\begin{definition}[Non-classical polynomials]
  For an integer $d \geq 0$,
  a function $P : \Fp^n \to \bbT$ is said to be a non-classical polynomial of degree at most $d$ (or simply a polynomial of degree at most $d$) if for all $x,y_1,\ldots,y_{d+1} \in \Fp^n$, it holds that
  \begin{align*}
    (D_{y_1} \cdots D_{y_{d+1}}P)(x) = 0. 
  \end{align*}
  The degree of $P$ is the smallest $d$ for which the above holds.
  A function $P : \Fp^n \to \bbT$ is said to be a classical polynomial of degree at most $d$ if it is a non-classical polynomial of degree at most $d$ whose image is contained in $\iota(\Fp)$.
\end{definition}
It is a direct consequence that a function $f : \Fp^n \to \bbC$ with $\|f\|_\infty \leq 1$ satisfies $\|f\|_{U^{d+1}} = 1$ if and only if $f = \sfe(P)$ for a (non-classical) polynomial $P : \Fp^n \to \bbT$ of degree at most $d$.

\begin{lemma}[Lemma 1.7 in~\cite{Tao:2011dw}]
  A function $P : \Fp^n \to \bbT$ is a polynomial of degree at most $d$ if and only if $P$ can be represented as
  \begin{align*}
    P(x_1,\ldots,x_n) = \alpha + \sum_{\substack{0\leq d_1,\ldots,d_n < p; h \geq 0: \\ 0 < \sum_i d_i \leq d - h(p-1)}} \frac{c_{d_1,\ldots,d_n,h}|x_1|^{d_1} \cdots |x_n|^{d_n}}{p^{h+1}} \bmod 1,
  \end{align*}
  for a unique choice of $c_{d_1,\ldots,dn,h} \in \{0,1,\ldots,p-1\}$ and $\alpha \in \bbT$.
  The element $\alpha$ is called the \emph{shift} of $P$,
  and the largest integer $h$ such that there exist $d_1,\ldots,d_n$ for which $c_{d_1,\ldots,d_n,h} \neq 0$ is called the \emph{depth} of $P$.
  Classical polynomials correspond to polynomials with 0 shift and 0 depth.
\end{lemma}

The degree and the depth of a polynomial $P$ is denoted by $\deg(P)$ and $\depth(P)$, respectively.
Also, for convenience of exposition, we will assume throughout this paper that the shifts of all polynomials are zero. 
This can be done without affecting any of the results in this work.
Hence, all polynomials of depth $h$ take values in $\bbU_{h+1}$.

\paragraph{Notations for polynomial sequences.}
Consider polynomials $P_1,\ldots,P_C:\Fp^n \to \bbT$ with respective degrees $d_1,\ldots,d_C$ and respective depths $h_1,\ldots,h_C$.
Let $\bfP = (P_1,\ldots,P_C)$, $\bfd = (d_1,\ldots,d_C)$, and $\bfh = (h_1,\ldots,h_C)$.
Then the \emph{degree} of $\bfP$ is  $\deg(\bfP) = \bfd$ and the \emph{depth} of $\bfP$ is $\depth(\bfP) = \bfh$.
Also, we say that $\bfP$ has \emph{degree less than $d$} if $d_i < d$ for any $i \in [C]$.
For a function $\Gamma:\bbT^C \to \bbC$,
we denote by $\Gamma \circ \bfP:\Fp^n \to \bbC$ the function with $(\Gamma \circ \bfP)(x) = \Gamma(P_1(x),\ldots,P_C(x))$ for any $x \in \Fp^n$.

%

\subsection{Polynomial factors and rank}

\begin{definition}[Factors]
  If $X$ is a finite set,
  then by a factor $\calB$,
  we mean a partition of $X$ into finitely many pieces called \emph{atoms}.
\end{definition}

A function $f : X \to \bbC$ is called \emph{$\calB$-measurable} if it is constant on atoms of $\calB$.
For any function $f : X \to \bbC$, we may define the conditional expectation
\begin{align*}
  \E[f \mid \calB](x) = \E[f(y) \mid y \in \calB(x)],
\end{align*}
where $\calB(x)$ is the unique atom in $\calB$ that contains $x$.
Note that $\E[f \mid \calB]$ is $\calB$-measurable.
A finite collection of functions $\phi_1, \ldots , \phi_C$ from $X$ to some other finite space $Y$ naturally define a factor $\calB = \calB(\phi_1,\ldots,\phi_C)$ whose atoms are sets of the form $\{x \mid (\phi_1(x),\ldots,\phi_C(x)) = (y_1,\ldots,y_C)\}$ for some $(y_1,\ldots,y_C) \in Y^C$.
By an abuse of notation we also use $\calB$ to denote the map $x \mapsto (\phi_1 (x), \ldots , \phi_C (x))$, thus also identifying the atom containing $x$ with $(\phi_1(x), \ldots , \phi_C(x))$. 
\begin{definition}[Polynomial factors]
  If $P_1, \ldots, P_C : \Fp^n \to \bbT$ is a sequence of polynomials,
  then the factor $\calB(P_1,\ldots,P_C)$ is called a polynomial factor.
\end{definition}

The \emph{complexity} of $\calB$, denoted $|\calB|$, is the number of defining polynomials $C$.
The degree of $\calB$ is the maximum degree among its defining polynomials $P_1, \ldots , P_C$.
If $P_1, \ldots , P_C$ are of depths $h_1, \ldots, h_C$, respectively, 
then $\|B\| = \prod_{i=1}^C p^{h_i+1}$ is called the \emph{order} of $\calB$.
Notice that the number of atoms of $\calB$ is bounded by $\|\calB\|$.
Next we need to define the notion of the rank of a polynomial or a polynomial factor.
\begin{definition}[Rank of a polynomial]
  Given a polynomial $P : \Fp^n \to \bbT$ and an integer $d > 1$,
  the \emph{$d$-rank} of $P$, denoted $\rank_d(P)$, is defined to be the smallest integer r such that there exist polynomials $Q_1,\ldots,Q_r : \Fp^n \to \bbT$ of degree at most $d-1$ and a function $\Gamma : \bbT^r \to \bbT$ satisfying $P(x) = \Gamma(Q_1(x), \ldots , Q_r(x))$. 
  If $d = 1$, then $1$-rank is defined to be $\infty$ if $P$ is non-constant and $0$ otherwise.
  The \emph{rank} of a polynomial $P : \Fp^n \to \bbT$ is its $\deg(P)$-rank.
\end{definition}
A high-rank polynomial of degree $d$ is, 
intuitively, a ``generic'' degree-$d$ polynomial.
There are no unexpected way to decompose it into polynomials of lower degrees.

Next, we will formalize the notion of a generic collection of polynomials.
Intuitively, it should mean that there are no unexpected algebraic dependencies among the polynomials.
\begin{definition}[Rank and Regularity]
  A polynomial factor $\calB$ defined by a sequence of polynomials $P_1,\ldots,P_C : \Fp^n \to \bbT$ with respective depths $h_1,\ldots,h_C$ is said to have rank $r$ if $r$ is the smallest integer for which there exist $(\lambda_1,\ldots,\lambda_C) \in \bbZ^C$ so that $(\lambda_1 \bmod p^{h_1+1},\ldots, \lambda_C \bmod p^{h_C+1}) \neq (0,\ldots,0)$ and the polynomial $Q = \sum_{i=1}^C \lambda_i P_i$ satisfies $\rank_d(Q) \leq r$ where $d = \max_i \deg(\lambda_i P_i)$.

  The rank of a polynomial sequence $\bfP$, denoted $\rank(\bfP)$, is the rank of the factor $\calB(\bfP)$.

  Given a polynomial factor $\calB$ and a function $r : \bbN \to \bbN$, 
  we say that $\calB$ is $r$-regular if $\calB$ is of rank at least $r(|\calB|)$.
\end{definition}

Note that, since $\lambda$ can be a multiple of $p$, rank measured with respect to $\deg(\lambda P)$ is not the same as rank measured with respect to $\deg(P)$. 
Thus for instance, if $\calB$ is the factor defined by a single polynomial $P$ of degree $d$ and depth $h$, then
\begin{align*}
  \rank(\calB) = \min \{\rank_d(P),\rank_{d-(p-1)}(pP), \ldots ,\rank_{d-h(p-1)}(p^hP)\}.
\end{align*}

Regular factors indeed do behave like generic collections of polynomials, and thus, given any factor $\calB$ that is not regular, it will often be useful to regularize $\calB$, that is, find a refinement $\calB'$ of $\calB$ that is regular up to our desires.
We distinguish between two kinds of refinements.

\begin{definition}[Semantic and syntactic refinements]\label{def:refinment}
  A polynomial factor $\calB'$ is called a \emph{syntactic refinement} of $\calB$, and denoted $\calB' \succeq_{\mathrm{syn}} \calB$, if the sequence of polynomials defining $\calB'$ extends that of $\calB$.
  It is called a \emph{semantic refinement}, and denoted $\calB' \succeq_{\mathrm{sem}} \calB$ if the induced partition is a combinatorial refinement of the partition induced by $\calB$.
  In other words, if for every $x,y \in \Fp^n$, $\calB'(x) = \calB'(y)$ implies $\calB(x) = \calB(y)$.
\end{definition}

The following lemma shows that every polynomial factor can be refined to be arbitrarily regular without increasing its complexity by more than a constant.
\begin{lemma}[Polynomial Regularity Lemma, Lemma~2.19 of~\cite{Bhattacharyya:2013ii}]\label{lem:polynomial-regularity-lemma}
  Let $d \in \bbN$ and $r : \bbN \to \bbN$ be a non-decreasing function.
  Then, there is a function $C^{(d,r)}_{\ref{lem:polynomial-regularity-lemma}} : \bbN \to \bbN$ with the following property.
  Suppose $\calB$ is a factor defined by polynomials $P_1, \ldots, P_C : \Fp^n \to \bbT$ of degree at most $d$.
  Then, there is an $r$-regular factor $\calB'$ consisting of polynomials $Q_1, \ldots, Q_{C'} : \Fp^n \to \bbT$ of degree at most $d$ such that $\calB' \succeq_{\mathrm{sem}} \calB$ and $C' \leq C^{(d,r)}_{\ref{lem:polynomial-regularity-lemma}}(C)$.

  Moreover, if $\calB$ itself is a refinement of some $\widehat{\calB}$ with rank at least $r(C') + C'$ and consists of polynomials,
  then additionally $\calB' \succeq_{\mathrm{syn}} \calB$. 
\end{lemma}

The first step towards showing that regular factors behave like generic collections of polynomials is to show that they form almost equipartitions.
\begin{lemma}[Size of atoms, Lemma 3.2 of~\cite{Bhattacharyya:2013ii}]\label{lem:size-of-atoms}
  Given $\epsilon > 0$, let $\calB$ be a polynomial factor of degree $d$, complexity $C$, and rank at least $r = r_{\ref{lem:size-of-atoms}}(\epsilon, d)$, defined by a polynomial sequence $P_1, \ldots, P_C : \Fp^n \to \bbT$.
  Suppose $b = (b_1,\ldots,b_C) \in \bbU_{\depth(P_1)+1} \times \cdots \times \bbU_{\depth(P_C)+1}$.
  Then
  \begin{align*}
    \Pr[\calB(x)=b] = \frac{1}{\|\calB\|} \pm  \epsilon.
  \end{align*}
  In particular, for $\epsilon < 1/\|\calB\|$, 
  $\calB(x)$ attains every possible value in its range and thus has $\| B\|$ atoms. 
\end{lemma}

Finally we state the regularity lemma, the basis of the higher-order Fourier analysis.
\begin{theorem}[Regularity Lemma, Theorem 4.4 of~\cite{Bhattacharyya:2012ud}]\label{the:regularity-lemma}
  Let $\zeta > 0$, $d \in \bbN$, and $\eta : \bbN \to \bbR^+$ be an arbitrary non-increasing function, 
  and let $r : \bbN \to \bbN$ be an arbitrary non-decreasing function.
  Let $\calB_0$ be a polynomial factor of degree d and complexity $C_0$.
  Then, there exists $C = C_{\ref{the:regularity-lemma}}(\eta, \zeta, C_0, d, r)$ with the following property.
  Every function $f : \Fp^n \to \{0, 1\}$ has a decomposition $f = f_1 + f_2 + f_3$ such that
  \begin{itemize}
  \item $f_1 = \E[f \mid \calB_1]$ for a polynomial factor $\calB_1 \succeq_{\mathrm{sem}} \calB_0$ of degree $d$ and complexity $C_1 \leq C$,
  \item $\|f_2\|_2 < \zeta$ and $\|f_3 \|_{U^{d+1}} < \eta(|\calB|)$,
  \item The functions $f_1$ and $f_1 + f_3$ have range $[0, 1]$; $f_2$ and $f_3$ have range $[-1, 1]$, and
  \item $\calB_1$ is $r$-regular.
  \end{itemize}
  Furthermore if $\rank(\calB_0) \geq r_{\ref{the:regularity-lemma}}(\eta, \zeta, C_0, d, r)$,
  then one can assume that $\calB_1 \succeq_{\mathrm{syn}} \calB_0$.
\end{theorem}

\subsection{Uniformity over linear forms}
A \emph{linear form} on $m$ variables is a vector $L = (\ell_1,\ldots,\ell_m) \in \Fp^m$.
We interpret it as a linear operator $L : (\Fp^n)^m \to \Fp^n$ given by $L(x_1,\ldots,x_m) = \sum_{i=1}^m \ell_i x_i$.


Let $\bfP = (P_1,\ldots,P_C)$ be a polynomial sequence and $\bfL = (L_1,\ldots,L_\ell)$ be a set of $\ell$ linear forms on $m$ variables.
Lemma~\ref{lem:size-of-atoms} says the distribution of $(P_i(x))_{i \in [C]}$ is close to uniform if the rank of $\bfP$ is high.
However, we also want to understand the distribution of $(P_i(L_j(x)))_{i \in [C], j \in [\ell]}$.
Unfortunately, the distribution could be far from uniform because of a trivial dependency among $L_1,\ldots,L_\ell$.
The following definition captures this dependency.
\begin{definition}\label{def:(d,h)-dependency}
  Given a set of linear forms $\bfL = (L_1,\ldots, L_\ell)$ on $m$ variables and $d, h \in \bbN$ such that $d > h(p - 1)$,
  the \emph{$(d,h)$-dependency set} of $\bfL$ is the set of tuples $(\lambda_1,\ldots,\lambda_\ell)$ with $\lambda_i \in \{0,\ldots,p^{h+1}-1\}$ for each $i \in [\ell]$ such that $\sum^\ell_{i=1} \lambda_iP(L_i(x_1,\ldots,x_m)) \equiv 0$ for every polynomial $P: \Fp^n \to \bbT$ of degree $d$ and depth $h$.
\end{definition}

The distribution of $(P_i(L_j(x)))_{i \in [C],j \in [\ell]}$ is only going to be supported on atoms with respect to the constraints imposed by dependency sets.
This is obvious:
if $P$ is a polynomial of degree $d$ and depth $h$,
$(\lambda_1,\ldots,\lambda_\ell)$ are in the $(d,h)$-dependency set of $\bfL = (L_1,\ldots,L_\ell)$,
and $P(L_j(x_1,\ldots,x_m)) = b_j$,
then $\sum_{j}\lambda_j b_j = 0$.
We call atoms with respect to this constraint for all $P_i$ in a factor \emph{consistent}.
Formally:
\begin{definition}[Consistency]\label{def:consitency}
  Let $\bfL$ be a set of $\ell$ linear forms.
  A sequence of elements $b_1,\ldots,b_\ell \in \bbT$ are said to be \emph{$(d,h)$-consistent} with $\bfL$ if $b_1,\ldots,b_\ell \in \bbU_{h+1}$ and for every tuple $(\lambda_1, \ldots , \lambda_\ell)$ in the $(d, h)$-dependency set of $\bfL$, it holds that $\sum^\ell_{i=1} \lambda_i b_i = 0$.

  Given vectors $\bfd = (d_1,\ldots,d_C) \in \bbN^C$ and $\bfh = (h_1,\ldots,h_C) \in \bbN^C$, 
  a sequence of vectors $b_1,\ldots,b_\ell \in \bbT^C$ are said to be \emph{$(\bfd,\bfh)$-consistent} with $\bfL$ if for every $i \in [C]$,
  the elements $b_{1,i},\ldots,b_{\ell,i}$ are $(d_i,h_i)$-consistent with $\bfL$.
  If $\calB$ is a polynomial factor, 
  the term $\calB$-consistent with $\bfL$ is a synonym for $(\bfd,\bfh)$-consistent with $\bfL$,
  where $\bfd$ and $\bfh$ are respectively the degree and depth of the polynomial sequence defining $\calB$.
\end{definition}

The following lemma says that,
given that the rank of $\bfP = (P_1,\ldots,P_C)$ is high enough,
the distribution of $(P_i(L_j(x)))_{i \in [C],j \in [\ell]}$ is close to uniform over atoms that is $\calB(\bfP)$-consistent with a set of linear forms $\bfL = (L_1,\ldots,L_\ell)$.
\begin{lemma}[Theorem~3.10 of~\cite{Bhattacharyya:2013ii}]\label{the:linear-form-equidistribution}
  Suppose $\epsilon > 0$.
  Let $\bfP$ be a sequence of $C$ polynomials with degree $\bfd = (d_1,\ldots,d_C)$ at most $d$, depth $\bfh = (h_1,\ldots,h_C)$, and $\rank(\bfP) \geq r_{\ref{the:linear-form-equidistribution}}(\epsilon, d)$.
  Let $\bfL = (L_1,\ldots,L_\ell)$ be a set of linear forms on $m$ variables.
  Suppose $b_1, \ldots , b_\ell \in \bbT^C$ are atoms of $\calB(\bfP)$ that are $\calB(\bfP)$-consistent with $\bfL$.
  Then
  \[
    \Pr_{x_1,\ldots,x_m}[\calB(L_j(x_1,\ldots,x_m)) = b_j\mbox{ for all } j \in [\ell]] = \frac{\prod_{i=1}^C|\Lambda_i|}{\|\calB\|^\ell} \pm \epsilon,
  \]
  where $\Lambda_i$ is the $(d_i,h_i)$-dependency set of $\bfL$.
\end{lemma}

Now we use Lemma~\ref{the:linear-form-equidistribution} to show that the Gowers norm of $\Gamma(\bfP) - \Gamma(\bfQ)$ is small if $\bfP$ and $\bfQ$ are of high rank and have the same degree and depth.
A point here is that $\bfP$ and $\bfQ$ can depend on different numbers of values.
\begin{lemma}\label{lem:Gamma-decides-gowers-norm}
  For any $\epsilon > 0$ and $C,d \in \bbN$,
  there exists $r = r_{\ref{lem:Gamma-decides-gowers-norm}}(\epsilon,C,d)$ with the following property.
  For any function $\Gamma:\bbT^C \to [0,1]$ and any polynomial sequences $\bfP$ and $\bfQ$ with complexity $C$, the same degree at most $d$, the same depth, and ranks at least $r$,
  we have $\|\Gamma \circ \bfP- \Gamma \circ \bfQ\|_{U^d} \leq \epsilon$.
\end{lemma}
\begin{proof}
  For a set $I \subseteq [d]$,
  let $L_I(x,y_1,\ldots,y_d) = x + \sum_{i \in I}y_i$.
  Let $\mu_\bfP$ and $\mu_{\bfQ}$ be distributions of the tuples $((\Gamma \circ \bfP)(L_I(x,y_1,\ldots,y_d)))_{I \subseteq [d]}$ and $((\Gamma \circ \bfQ)(L_I(x,y_1,\ldots,y_d)))_{I \subseteq [d]}$, respectively.
  We consider the statistical distance between $\mu_\bfP$ and $\mu_\bfQ$.
  Let $\bfL = (L_I)_{I \subseteq [d]}$ and $\ell = |\bfL| = 2^d$.
  We set $r_{\ref{lem:Gamma-decides-gowers-norm}}(\epsilon, C, d) = r_{\ref{the:linear-form-equidistribution}}(\epsilon^{2^d}/(2p^{Cd\ell}), d)$,

  Then for any atoms $\{b_I\}_{I \subseteq [d]}$ in $\calB(\bfP)$ (and hence in $\calB(\bfQ)$) that are $(\bfd,\bfh)$-consistent,
  we have 
  \begin{align*}
    \Pr_{x,y_1,\ldots,y_d}[\bfP(L_I(x,y_1,\ldots,y_d)) = b_I\mbox{ for all } I \subseteq [d]] & = \frac{\prod_{i=1}^C|\Lambda_i|}{\|\calB\|^\ell} \pm \frac{\epsilon^{2^d}}{2p^{Cd\ell}}, \quad \mbox{and} \\
    \Pr_{x,y_1,\ldots,y_d}[\bfQ(L_I(x,y_1,\ldots,y_d)) = b_I\mbox{ for all } I \in [d]] & = \frac{\prod_{i=1}^C|\Lambda_i|}{\|\calB\|^\ell} \pm \frac{\epsilon^{2^d}}{2p^{Cd\ell}}.
  \end{align*}

  Since the number of atoms in $\calB(\bfP)$ and $\calB(\bfQ)$ are at most $p^{dC}$,
  we have $\dtv(\mu_{\bfP},\mu_{\bfQ}) \leq \epsilon^{2^d}$.

  Recall that the Gowers norm of $\Gamma \circ \bfP - \Gamma \circ \bfQ$ can be written as follows.
  \[
    \|\Gamma \circ \bfP- \Gamma \circ \bfQ\|_{U^d}^{2^d} = 
    \left| \E_{x,y_1,\ldots,y_d}  \prod_{I \subseteq [d]} \left(\Gamma (\bfP(x + \sum_{i \in I}y_i)) - \Gamma(\bfQ(x + \sum_{i \in I}y_i))\right) \right|.
  \]
  Over the choice of $x,y_1,\ldots,y_d$,
  the probability that $\bfP(x + \sum_{i \in I}y_i)$ and $\bfQ(x + \sum_{i \in I}y_i)$ have different values for some $I \subseteq [d]$ is at most $\epsilon^{2^d}$.
  Since the range of $\Gamma$ is $[0,1$], 
  we have $\|\Gamma \circ \bfP - \Gamma \circ \bfQ\|_{U^d}^{2^d} \leq \epsilon^{2^d}$, from which the lemma follows.
\end{proof}

\subsection{Properties of affine embeddings}

It is not difficult to see that,
for affine-invariant properties, local testability has an equivalent non-algorithmic definition through the distribution of restrictions to affine subspaces.
The following proposition is essentially due to~\cite{Bhattacharyya:2010gb}.
\begin{proposition}\label{pro:canonical-tester}
  An affine-invariant property $\calP$ is locally testable if and only if,
  for every $\epsilon > 0$,
  there exist a constant $m$ and a set $\calV \subseteq \{\Fp^m \to \bit\}$ with the following property.
  For any function $f:\Fp^n \to \bit$,
  over a random affine embedding $A:\Fp^m \to \Fp^n$,
  \begin{itemize}
  \item we have $\Pr_A[f \circ A \in \calV] \geq 2/3$ if $f \in \calP$, and 
  \item we have $\Pr_A[f \circ A \not \in \calV] \geq 2/3$ if $f$ is $\epsilon$-far from $\calP$.
  \end{itemize}
\end{proposition}

Using Proposition~\ref{pro:canonical-tester}, 
the following lemma is shown in~\cite{Hatami:2013ux}, 
\begin{lemma}\label{lem:bad-affine-embedding}
  Let $\epsilon > 0$, $C \in \bbN$, $d \in \bbN$, and $r \in \bbN$.
  Let $\bfd = (d_1,\ldots,d_C) \in \bbN^C$, $\bfh = (h_1,\ldots,h_C)\in \bbN^C$ with $d_i < d$ and $h_i < d_i$ for every $i \in [C]$.
  Suppose $m \geq m_{\ref{lem:bad-affine-embedding}}(\epsilon,C,d,r)$.
  Then for every sequence $\bfP$ of $C$ polynomials $P_1,\ldots,P_C: \Fp^n \to \bbT$ with $\deg(\bfP) = \bfd$, $\depth(\bfP) = \bfh$, and $\rank(\bfP) \geq r$,
  a random affine embedding $A:\Fp^m \to \Fp^n$ satisfies  
  \begin{align*}
    \Pr[\deg(P_i \circ A) < d_i\mbox{ for some }i \in [C] \vee
    \depth(P_i \circ A) < h_i\mbox{ for some }i \in [C] \vee
    \rank(\bfP \circ A) < r] < \epsilon.
  \end{align*}
\end{lemma}


The following lemma gives a behavior of the Gowers norm through affine embeddings.
\begin{lemma}[Claim~4.1 of~\cite{Hatami:2013ux}]\label{lem:gowers-norm-after-affine-embedding}
  Given $\epsilon > 0$ and $d \in \bbN$, suppose $m \geq m_{\ref{lem:gowers-norm-after-affine-embedding}}(\epsilon,d)$.
  Let $f:\bbF^n \to [-1,1]$ be a function.
  With probability at least $99/100$ over the choice of a random affine embedding $A:\Fp^m \to \Fp^n$,
  we have $\|f \circ A\|_{U^d} \leq \|f\|_{U^d} + \epsilon$.
\end{lemma}



\section{Satisfying Regularity-Instances by Small Perturbations}\label{sec:small-perturbation}
Let $I = (\gamma,\Gamma,C,d,\bfd,\bfh,r)$ be a regularity-instance.
Suppose that a function $f$ can be decomposed as $f(x) = \widetilde{\Gamma}(\bfP(x)) + \Upsilon(x)$,
where $\bfP$ is a sequence of $C$ polynomials with $\deg(\bfP) = \bfd$, $\depth(\bfP) = \bfh$, and $\rank(\bfP) \geq r$,
$\widetilde{\Gamma}$ is a function close to $\Gamma$,
and $\Upsilon$ has Gowers norm slightly larger than $\gamma$.
In this section, we show that such a function $f$ can be made satisfy $I$ by a small perturbation.
Formally, we show the following.
\begin{lemma}\label{lem:small-perturbation}
  For any $\gamma,\epsilon > 0$ and $d \in \bbN$,
  there exist $\tau = \tau_{\ref{lem:small-perturbation}}(\gamma,\epsilon,d)$ and $r_{\ref{lem:small-perturbation}}(\gamma,\epsilon,C,d)$ with the following property.
  Let $I = (\gamma,\Gamma,C,d,\bfd,\bfh,r)$ be a regularity-instance with $r \geq r_{\ref{lem:small-perturbation}}(\gamma,\epsilon,C,d)$.
  Suppose that a function $f:\Fp^n \to \bit$ can be expressed as 
  \[
    f(x) = \widetilde{\Gamma}(\bfP(x)) + \Upsilon(x),
  \]
  where
  \begin{itemize}
  \item $\bfP$ is a polynomial sequence with $\deg(\bfP) = \bfd$, 
  $\depth(\bfP) = \bfh$, and
  $\rank(\bfP) \geq r$,
  \item $\widetilde{\Gamma}:\prod_{i=1}^C \bbU_{h_i+1} \to [0,1]$ is a function with $\|\Gamma - \widetilde{\Gamma}\|_\infty \leq \tau$, 
  where $\bfh = (h_1,\ldots,h_C)$, and
  \item $\Upsilon:\Fp^n \to [-1,1]$ is a function with $\|\Upsilon\|_{U^d} \leq \gamma + \tau$.
  \end{itemize}
  Then, $f$ is $\epsilon$-close to satisfying $I$.
\end{lemma}


Let $\Upsilon'(x) = \widetilde{\Gamma}(\bfP(x)) - \Gamma(\bfP(x)) + \Upsilon(x)$.
Then, 
we can switch the structured part of $f$ to $\Gamma(\bfP(x))$ by expressing $f$ as $f(x) = \Gamma(\bfP(x)) + \Upsilon'(x)$.
The following claim shows that the Gowers norm of $\Upsilon'$ is not much larger than that of $\Upsilon$.
\begin{claim}\label{cla:small-paturbation-adjust-Gamma}
  Suppose $r_{\ref{lem:small-perturbation}}(\gamma,\epsilon,C,d) \geq  r_{\ref{lem:size-of-atoms}}(\tau / p^{dC})$.
  Then $\|\Upsilon'\|_{U^d} \leq \gamma + (2\tau)^{1/2^d}$.
\end{claim}
\begin{proof}
  Since $\|\Upsilon\|_{U^d} \leq \gamma$,
  it suffices to bound the Gowers norm of $\widetilde{\Gamma} \circ \bfP - \Gamma \circ \bfP$.
  Since $r \geq r_{\ref{lem:size-of-atoms}}(\tau / p^{dC})$ and $\|\Gamma\|_{\infty} \leq 1$,
  by Lemma~\ref{lem:size-of-atoms},
  we have
  $\|\widetilde{\Gamma} \circ \bfP - \Gamma \circ \bfP\|_1 \leq \tau + \tau/ p^{dC} \cdot \|\calB\| \leq 2\tau$.
  By Lemma~\ref{lem:bound-gowers-norm-by-l1-norm},
  we have $\|\widetilde{\Gamma} \circ \bfP - \Gamma \circ \bfP\|_{U^d} \leq (2\tau)^{1/2^d}$.
\end{proof}
In what follows, we assume $r_{\ref{lem:small-perturbation}}(\gamma,\epsilon,C,d) \geq  r_{\ref{lem:size-of-atoms}}(\tau / p^{dC})$.
From Claim~\ref{cla:small-paturbation-adjust-Gamma},
the pseudo-random part has Gowers norm at most $\gamma + (2\tau)^{1/2^d}$.
To make the Gowers norm at most $\gamma$,
we now construct a function $g$ as follows.
For each point $x \in \Fp^n$,
we decide the value of $g$ by tossing two coins.
The first coin comes up heads with probability $1-\delta$ and tails with probability $\delta$, where $\delta$ is a parameter chosen later.
If the first coin comes up heads, we set $g(x) = f(x)$.
If the first coin comes up tails, we toss the second coin.
The second coin comes up heads with probability $\beta$ and tails with probability $1-\beta$,
where $\beta = \Gamma(b)$ for the atom $b$ of $\calB(\bfP)$ corresponding to $x$.
We set $g(x) = 1$ if the second coin comes up heads and set $g(x) = 0$ otherwise.
\begin{claim}\label{cla:small-paturbation-distance-to-g}
  For sufficiently large $n$,
  we have $\|f - g\|_1 \leq 2\delta $ with probability $1 - o_n(1)$.
\end{claim}
\begin{proof}
  Note that $\E_g \|f-g\|_1 \leq \delta$ and $\Var\|f - g\|_1 \leq p^{-n}$.
  Hence by Chebyshev's inequality,
  $\|f - g\|_1 \leq \delta + o_n(1) \leq 2\delta$ with probability $1 - o_n(1)$.
\end{proof}

\begin{claim}\label{cla:small-paturbation-decrease-gowers-norm}
  For sufficiently large $n$,
  we have $\|g - \Gamma \circ \bfP\|_{U^d} \leq (1 - \delta/3)  \|f - \Gamma \circ \bfP\|_{U^d}$ with probability at least $\delta / 2$.
\end{claim}
\begin{proof}
  For a function $h$, we define $\gamma_h = \|h - \Gamma \circ \bfP\|_{U^d}^{2^d}$.
  The expected Gowers norm of $g - \Gamma \circ \bfP$ is
  \[
    \E_g \gamma_g = \E_{x,y_1,\ldots,y_d \in \Fp^n} \E_g\prod_{I \subseteq [d]} (g(x+\sum_{i \in I}y_i) - (\Gamma \circ \bfP)(x + \sum_{i \in I}y_i) ).
  \]
  If all $y_1,\ldots,y_d$ are linearly independent, 
  each term in the product becomes independent.
  Since this happens with probability at least $1 - p^{d-n}$,
  we have 
  \[
    \E_g \gamma_g = \E_{x,y_1,\ldots,y_d \in \Fp^n}\prod_{I \subseteq [d]}  \E_g[g(x+\sum_{i \in I}y_i) - (\Gamma \circ \bfP)(x + \sum_{i \in I}y_i) ]+ o_n(1).
  \]
  For $\bfx = (x,y_1,\ldots,y_d)$,
  we define $f'(\bfx) = \prod_{I \subseteq [d]} (f(x+\sum_{i \in I}y_i) - (\Gamma \circ \bfP)(x + \sum_{i \in I}y_i) )$.
  Similarly we define $g'(\bfx) = \prod_{I \subseteq [d]} (\E_g g(x+\sum_{i \in I}y_i) - (\Gamma \circ \bfP)(x + \sum_{i \in I}y_i) )$.
  Then, we have 
  \[
    \gamma_f = \E_{\bfx \in (\Fp^n)^{d+1}}f'(\bfx) \quad \text{and} \quad \E_g \gamma_g = \E_{\bfx \in (\Fp^n)^{d+1}}g'(\bfx) + o_n(1).
  \]

  Now we consider the difference between $\E_\bfx f'(\bfx)$ and $ \E_\bfx \E_g g'(\bfx)$.

  Let $V \subseteq \Fp^n$ be the set of points $x \in \Fp^n$ for which the corresponding first coin comes up tails.
  Then, $g'(\bfx) = f'(\bfx)$ if no point in $\bfx$ belongs to $V$ and $g'(\bfx) = 0$ otherwise.
  Hence, the probability that $g'(\bfx) = 0$ is $q$, where $q = 1 - (1-\delta)^{2^d} \geq \delta$,
  which means that $\E_g g'(\bfx) \leq (1-\delta)f'(\bfx)$ holds for any $\bfx \in (\Fp^n)^{d+1}$.
  Thus $\E_g \E_x g'(\bfx) \leq (1-\delta)\gamma_f$.
  From Markov's inequality,
  the probability that $\E_x g'(\bfx) \geq (1-\delta/2)\gamma_f$ is at most $(1-\delta)/(1-\delta/2)$.
  It means that, with probability at least $1 - (1-\delta)/(1-\delta/2) = \delta/2(1-\delta) \geq \delta/2$,
  the Gowers norm of $g - \Gamma \circ \bfP$ is at most $(1-\delta/2)f'(\bfx) + o_n(1) \leq (1-\delta/3)f'(\bfx)$.
\end{proof}

From the probabilistic argument,
there is a function $g$ satisfying both consequences of Claims~\ref{cla:small-paturbation-distance-to-g} and~\ref{cla:small-paturbation-decrease-gowers-norm}:
\begin{corollary}\label{cor:small-paturbation-get-g}
  For sufficiently large $n$,
  there exists $g$ such that $\|f - g\|_1 \leq 2\delta$ and $\|g - \Gamma\circ \bfP\|_{U^d} \leq (1-\delta/3)\|f - \Gamma \circ \bfP\|_{U^d}$.
\end{corollary}

We choose $2\delta \leq \epsilon$ and choose $\tau$ so that $(\gamma + (2\tau)^{1/2^d})(1 - \delta/3) \leq \gamma$.
Then Lemma~\ref{lem:small-perturbation} follows from Claims~\ref{cla:small-paturbation-adjust-Gamma} and Corollary~\ref{cor:small-paturbation-get-g}.

\section{Regularity-Instances are Locally Testable}\label{sec:regularity-instance->testable}
In this section, we show that the property of satisfying a regularity-instance is locally testable.
Throughout this section,
we fix the proximity parameter $\epsilon$ and the regularity-instance $I = (\gamma,\Gamma,C,d,\bfd,\bfh,r)$ with $r \geq r_{\ref{the:regularity-instance->testable}}(\gamma,\epsilon,C,d)$ for some $r_{\ref{the:regularity-instance->testable}}(\gamma,\epsilon,C,d)$ defined later.

Our $\epsilon$-tester for the property of satisfying $I$ is very simple:
We choose $\delta = \delta(\gamma,\epsilon,C,d)$ small enough and $m = m(\gamma,\epsilon,C,d)$ large enough (these parameters are used throughout this section).
Given a function $f:\Fp^n \to \bit$,
we choose a random affine embedding $A:\Fp^m \to \Fp^n$.
Then, we accept if $f\circ A$ is $\delta$-close to satisfying $I$ and reject if $f \circ A$ is $\delta$-far from satisfying $I$.
Clearly, the number of queries only depends on $\epsilon$ and $I$.

It is easy to show as follows that the tester accepts with high probability when $f$ satisfies $I$.
\begin{lemma}\label{lem:regularity-instance-completeness}
  Suppose $m \geq m_{\ref{lem:regularity-instance-completeness}}(\gamma,\epsilon,C,d)$ and $r \geq r_{\ref{lem:regularity-instance-completeness}}(\gamma,\epsilon,C,d)$.
  Then for any function $f:\Fp^n \to \bit\;(n \geq m)$ satisfying $I$,
  over the choice of an affine embedding $A:\Fp^m \to \Fp^n$,
  $f \circ A$ is $\delta$-close to satisfying $I$ with probability at least $2/3$.
\end{lemma}
\begin{proof}
  Since $f$ satisfies $I$,
  $f$ can be written as $f(x) = \Gamma(\bfP(x)) + \Upsilon(x)$,
  where $\bfP$ is a polynomial sequence with complexity $C$, $\deg(\bfP) = \bfd$, $\depth(\bfP) = \bfh$, and $\rank(\bfP) \geq r$,
  and $\|\Upsilon\|_{U^d} \leq \gamma$.
  Note that $f(Ax) = \Gamma(\bfQ(x)) + \Upsilon(Ax)$ holds, where $\bfQ = \bfP \circ A$.
  We choose $m_{\ref{lem:regularity-instance-completeness}}(\gamma,\epsilon,C,d)$ as
  \[
    m_{\ref{lem:regularity-instance-completeness}}(\gamma,\epsilon,C,d)
    \geq
    \max\{m_{\ref{lem:bad-affine-embedding}}(1/100,C,d,r), 
    m_{\ref{lem:gowers-norm-after-affine-embedding}}(\tau_{\ref{lem:small-perturbation}}(\gamma,\delta,d), d)
    \}.
  \]
  Then by Lemma~\ref{lem:bad-affine-embedding},
  $\deg(\bfQ) = \bfd$, $\depth(\bfQ) = \bfh$, and $\rank(\bfQ) \geq r$ holds with probability at least $99/100$.
  Also by Lemma~\ref{lem:gowers-norm-after-affine-embedding},
  we have $\|\Upsilon \circ A\|_{U^d} \leq \gamma + \tau_{\ref{lem:small-perturbation}}(\gamma,\delta,d)$ with probability at least $99/100$.
  Hence, with probability at least $2/3$, both of these happen.
  By choosing $r_{\ref{lem:regularity-instance-completeness}}(\gamma,\epsilon,C,d) \geq r_{\ref{lem:small-perturbation}}(\gamma,\delta,C,d)$,
  such a function is indeed $\delta$-close to satisfying the regularity-instance $I$ from Lemma~\ref{lem:small-perturbation}.
\end{proof}
The following lemma handles the case that $f$ is $\epsilon$-far.
Its proof is given in Section~\ref{subsec:regularity-instance-soundness},
\begin{lemma}\label{lem:regularity-instance-soundness}
  Suppose 
  $\delta \leq \delta_{\ref{lem:regularity-instance-soundness}}(\gamma,\epsilon,C,d)$, $m \geq m_{\ref{lem:regularity-instance-soundness}}(\gamma,\epsilon,C,d)$ and $r \geq r_{\ref{lem:regularity-instance-soundness}}(\gamma,\epsilon,C,d)$.
  Then, for any function $f:\Fp^n \to \bit\;(n \geq m)$ that is $\epsilon$-far from satisfying $I$,
  over the choice of an affine embedding $A:\Fp^m \to \Fp^n$,
  $f \circ A$ is $\delta$-far from satisfying $I$ with probability at least $2/3$.
\end{lemma}
Now we establish Theorem~\ref{the:regularity-instance->testable} by choosing $\delta \leq \delta_{\ref{lem:regularity-instance-soundness}}(\gamma,\epsilon,C,d)$, $m \geq \max\{m_{\ref{lem:regularity-instance-completeness}}(\gamma,\epsilon,C,d),m_{\ref{lem:regularity-instance-soundness}}(\gamma,\epsilon,C,d)\}$, and $r \geq \max\{r_{\ref{lem:regularity-instance-completeness}}(\gamma,\epsilon,C,d), r_{\ref{lem:regularity-instance-soundness}}(\gamma,\epsilon,C,d)\}$.

\subsection{Proof of Lemma~\ref{lem:regularity-instance-soundness}}\label{subsec:regularity-instance-soundness}
In this section, we prove Lemma~\ref{lem:regularity-instance-soundness}.
Suppose for contradiction that,
with probability more than $1/3$, $f \circ A$ is $\delta$-close to satisfying a regularity-instance $I$, that is, $f(Ax) = \Gamma(\bfP'(x)) + \Upsilon'(x) + \Delta'(x)$ for some polynomial sequence $\bfP' = (P_1,\ldots,P_C)$ on $m$ variables with $\deg(\bfP') = \bfd$, $\depth(\bfP') = \bfh$, and $\rank(\bfP') \geq r$, a function $\Upsilon':\Fp^m \to [-1,1]$ with $\|\Upsilon'\|_{U^d} \leq \gamma$, and a function $\Delta':\Fp^m \to \bit$ with $\|\Delta'\|_1 \leq \delta$.
We note that the range of $f \circ A - \Delta'$ is $\bit$.
Let $\rho = \rho(\gamma,\epsilon, d)$ be a parameter that will be determined later.
We set parameters $\eta_\bfR:\bbN \to \bbR$, $\zeta_\bfR \in \bbR^+$ and $r_\bfR:\bbN \to \bbN$ so that $\eta_\bfR(D) \leq \rho/2$ for any $D \in \bbN$,
$\zeta_\bfR = \rho^{2^d}/2$, 
and $r_\bfR(D) > r_{\overline{\bfR'}}(D) + D$ for any $D \in \bbN$,
where $r_{\overline{\bfR'}}:\bbN \to \bbN$ is a function defined later ($r_{\overline{\bfR'}}$ will depend only on $d$ and $p$).
We apply Theorem~\ref{the:regularity-lemma} to $f$ with parameters $\eta_\bfR$, $\zeta_\bfR$, $0$ (in place of $C_0$), $d$, and $r_\bfR$.
Then, the function $f$ can be decomposed as follows.
\begin{claim}\label{cla:properties-of-R}
  The function $f$ is decomposed as $f = f_1 + f_2 + f_3$ with the following properties.
  \begin{itemize}
  \item $f_1 = \Sigma \circ \bfR$ for some function $\Sigma: \bbT^{|\bfR|} \to [0,1]$ and some polynomial sequence $\bfR = (R_1,\ldots,R_{|\bfR|})$ on $n$ variables with size at most $C_{\ref{the:regularity-lemma}}(\eta_\bfR,\zeta_\bfR,0,d,r_\bfR)$, degree less than $d$, and rank at least $r_{\bfR}(|\bfR|)$,
  \item $\|f_2\|_{U^d} \leq \rho$, and
  \item $\|f_3\|_{U^d} \leq \rho$.
  \end{itemize}
\end{claim}
\begin{proof}
  The first and the third properties are direct consequences of Theorem~\ref{the:regularity-lemma}.
  The second property also holds as $\|f_2\|_{U^d} \leq \|f_2\|_2^{1/2^d} \leq \zeta_\bfR^{1/2^d} \leq \rho$ from Lemma~\ref{lem:bound-gowers-norm-by-l1-norm}.
\end{proof}

Define $\bfR' = (R'_1,\ldots,R'_{|\bfR|})$ by $R'_i = R_i \circ A$ for each $i \in [|\bfR|]$.
It is shown in Claim~4.1 of~\cite{Hatami:2013ux} that many properties of $f_1$, $f_2$, and $f_3$ are preserved in $f_1 \circ A$, $f_2\circ A$, and $f_3 \circ A$, respectively.
In our scenario, we have the following.
\begin{claim}\label{cla:good-case}
  Suppose $m_{\ref{lem:regularity-instance-soundness}}(\gamma,\epsilon,C,d) \geq m_{\ref{cla:good-case}}(\eta_\bfR,d,r_{\bfR})$.
  Then the following events hold with probability at least $99/100$.
  \begin{itemize}
  \item $\deg(\bfR') = \deg(\bfR)$, $\depth(\bfR') = \depth(\bfR)$, and $\calB(\bfR')$ is $r_\bfR$-regular.
  \item We have $\|f_2 \circ A\|_2 \leq 2\zeta_\bfR$ and $\|f_3 \circ A\|_{U^d} \leq 2\eta_\bfR(|\bfR'|)$.
  \end{itemize}
\end{claim}
In what follows, we assume that $m_{\ref{lem:regularity-instance-soundness}}(\gamma,\epsilon,C,d) \geq m_{\ref{cla:good-case}}(\eta_\bfR,d,r_{\bfR})$ and the consequence of Claim~\ref{cla:good-case} actually holds (We have such a situation with probability at least $1/3-1/100$).
Hence we have $\|f_2 \circ A\|_{U^d} \leq (2\zeta_\bfR)^{1/2^d} \leq \rho$ and $\|f_3 \circ A\|_{U^d} \leq 2\eta_\bfR(|\bfR|) \leq \rho$.


Now $f(Ax)$ can be expressed in the following two ways.
\begin{align*}
  f(Ax) = \Sigma(\bfR'(x)) + f_2(Ax) + f_3(Ax) = \Gamma(\bfP'(x)) + \Upsilon'(x) + \Delta'(x).
\end{align*}
We further refine the factor $\calB(\bfR' \cup \bfP')$.
We set $r_{\overline{\bfR'}}:\bbN \to \bbN$ so that $r_{\overline{\bfR'}}(D) \geq r_{\ref{lem:size-of-atoms}}(1/(2p^{dD}),d)$ for any $D \in \bbN$.
Now we apply Lemma~\ref{lem:polynomial-regularity-lemma} to find an $r_{\overline{\bfR'}}$-regular refinement of the factor $\calB(\bfP' \cup \bfR')$.
Since $\calB(\bfR')$ is $r_\bfR$-regular with $r_\bfR(D) \geq r_{\overline{\bfR'}}(D) + D$ for any $D \in \bbN$, we obtain an extension $\overline{\bfR'} = \bfR' \cup \bfS' = (\overline{R'_i})_{i \in [|\overline{\bfR'}|]}$ of $\bfR'$ for some polynomial sequence $\bfS'$ of degree less than $d$.

Since $\calB(\overline{\bfR'})$ is a refinement of $\calB(\bfP' \cup \bfR')$,
for each $i \in [C]$,
there exists some function $\Gamma_i:\prod_{i = 1}^{|\overline{\bfR'}|}\bbU_{\depth(\overline{R'_i})+1} \to \bbU_{h_i+1}$ such that $P'_i = \Gamma_i(\overline{\bfR'})$.
Hence,
\[
  \Sigma(\bfR'(x)) + f_2(Ax) + f_3(Ax) = \Gamma(\Gamma_1(\overline{\bfR'}(x)), \ldots,\Gamma_{C}(\overline{\bfR'}(x)) + \Upsilon'(x) + \Delta'(x).
\]

Since $\bfR'$ is a subsequence of $\overline{\bfR'}$, $f_2(Ax) + f_3(Ax) - \Upsilon'(x) - \Delta'(x)$ is measurable with respect to the factor $\calB(\overline{\bfR'})$.
Thus, we can write $f_2(Ax) + f_3(Ax) - \Upsilon'(x) - \Delta'(x) = \Phi(\overline{\bfR'}(x))$ for some function $\Phi:\prod_{i = 1}^{|\overline{\bfR'}|}\bbU_{\depth(\overline{R'_i})+1} \to [-1,1]$.
The range of $\Phi$ is $[-1,1]$ since the ranges of $\Sigma$ and $\Gamma$ are $[0,1]$.
Now we have
\[
  \Sigma(\bfR'(x)) + \Phi(\overline{\bfR'}(x)) = \Gamma(\Gamma_1(\overline{\bfR'}(x)), \ldots,\Gamma_{C}(\overline{\bfR'}(x))).
\]

Since $r_{\overline{\bfR'}}(|\calB(\overline{\bfR'})|) \geq r_{\ref{lem:size-of-atoms}}(1/(2p^{d|\calB(\overline{\bfR'})|}),d) \geq r_{\ref{lem:size-of-atoms}}(1/(2\|\calB(\overline{\bfR'})\|),d)$,
by Lemma~\ref{lem:size-of-atoms},
the tuple $\overline{\bfR'}(x)$ acquires every value in its range.
Thus for all $\bfa \in \prod_{i = 1}^{|\bfR'|}\bbU_{\depth(R'_i)+1}$ and $\bfb \in \prod_{i = 1}^{|\bfS'|}\bbU_{\depth(S'_i)+1}$,
we have the identity
\[
  \Sigma(\bfa) + \Phi(\bfa,\bfb) = \Gamma(\Gamma_1(\bfa,\bfb), \ldots,\Gamma_{C}(\bfa,\bfb)).
\]

Let $A^+:\Fp^n \to \Fp^m$ be any affine transformation with $A^+ A = I_m$.
We define a polynomial sequence $\bfS = (S_1,\ldots,S_{|\bfS'|})$ on $n$ variables by setting $S_i(x) = S'_i \circ A^+$ for each $i \in [|\bfS'|]$.
We set $\overline{\bfR} = \bfR \cup \bfS$.
We define a polynomial sequence $\bfP = (P_1,\ldots,P_C)$ on $n$ variables by setting $P_i(x) = \Gamma_i(\overline{\bfR}(x))$ for each $i \in [C]$.
Note that $P'_i = P_i \circ A$ for each $i \in [C]$.
We have
\begin{align*}
  \Sigma(\bfR(x)) + \Phi(\overline{\bfR}(x)) = \Gamma(\bfP(x)).
\end{align*}

Most properties of $\overline{\bfR'}$ are preserved in $\overline{\bfR}$ as shown in the following claim.
\begin{claim}\label{cla:property-of-overline-bfR}
  We have $\deg(\overline{\bfR}) = \deg(\overline{\bfR'})$,
  $\depth(\overline{\bfR}) = \depth(\overline{\bfR'})$,
  and $\rank(\overline{\bfR}) \geq \rank(\overline{\bfR'})$.
\end{claim}
\begin{proof}
  We have $\deg(\bfR) = \deg(\bfR')$ and $\depth(\bfR) = \deg(\bfR')$ from Claim~\ref{cla:good-case}.
  Also, we have $\deg(\bfS) = \deg(\bfS')$ and $\depth(\bfS) = \depth(\bfS')$ since $\bfS' = \bfS \circ A$ and $\bfS = \bfS' \circ A^+$ and affine transformation does not increase degree and depth.
  Hence, $\deg(\overline{\bfR}) = \deg(\overline{\bfR'})$ and $\depth(\overline{\bfR}) = \depth(\overline{\bfR'})$ hold.
  Since $\overline{\bfR'} = \overline{\bfR} \circ A$ and affine transformation does not increase rank,
  $\rank(\overline{\bfR}) \geq \rank(\overline{\bfR'})$.
\end{proof}

The following lemma is useful to analyze the property of $\bfP$.
In the following lemma, symbols $\bfP$ and $\Gamma$ are nothing to do with those in the current context.
\begin{lemma}[Theorem~4.1 of~\cite{Bhattacharyya:2013ii}]\label{lem:degree-non-increase}
  For an integer $d > 0$, let $\bfP = (P_1,\ldots,P_C)$ be a polynomial sequence of degree at most $d$ and rank at least $r_{\ref{lem:degree-non-increase}}(d)$, and let $\Gamma: \bbT^C \to \bbT$ be a function.
  Then, for every polynomial sequence $\bfQ = (Q_1,\ldots,Q_C)$ with $\deg(Q_i) \leq \deg(P_i)$ and $\depth(Q_i) \leq \depth(P_i)$ for all $i \in [C]$, 
  it holds that $\deg(\Gamma \circ \bfQ) \leq \deg(\Gamma \circ \bfP)$.
\end{lemma}
Now we come back to the proof of Lemma~\ref{lem:regularity-instance-soundness}.
We have the following.
\begin{claim}\label{cla:property-of-bfP}
  If $r_{\ref{lem:regularity-instance-soundness}}(\gamma,\epsilon,C,d) \geq r_{\ref{lem:degree-non-increase}}(d)$, then
  we have $\deg(\bfP) = \bfd$,
  $\depth(\bfP) = \bfh$,
  and $\rank(\bfP) \geq r$. 
\end{claim}
\begin{proof}
  Note that $\bfP' = \bfP \circ A$.
  Since affine transformation does not increase degree and rank,
  we have $\deg(\bfP) \geq \deg(\bfP') = \bfd$,
  $\depth(\bfP) \geq \depth(\bfP') = \bfh $,
  and $\rank(\bfP) \geq \rank(\bfP') = r$.
  Since $\deg(\overline{\bfR}) = \deg(\overline{\bfR'})$ and $\depth(\overline{\bfR}) = \depth(\overline{\bfR'})$ from Claim~\ref{cla:good-case},
  by Lemma~\ref{lem:degree-non-increase},
  we have $\deg(\bfP) = \deg(\Gamma \circ \overline{\bfR}) \leq \deg(\Gamma \circ \overline{\bfR'}) = \deg(\bfP') = \bfd$.
  Hence, we have $\deg(\bfP) = \bfd$.
  Also for each $i \in [C]$,
  since $P_i = \Gamma_i(\overline{\bfR})$ and the range of $\Gamma_i$ is $\bbU_{h_i+1}$,
  we have $\depth(P_i) \leq h_i$.
  Hence, we have $\depth(\bfP) = \bfh$.
\end{proof}
Claim~\ref{cla:property-of-bfP} in particular says that $\Gamma(\bfP(x))$ is a function  satisfying the regularity-instance $I$.
In what follows, we assume $r_{\ref{lem:regularity-instance-soundness}}(\gamma,\epsilon,C,d) \geq r_{\ref{lem:degree-non-increase}}(d)$.

We now want to show that $f(x)$ and $\Gamma(\bfP(x))$ are close.
Recall that $\Gamma(\bfP(x)) = f(x) - f_2(x) - f_3(x) + \Phi(\overline{\bfR}(x))$.
We already know that $\|f_2\|_{U^d}$ and $\|f_3\|_{U^d}$ are small from Claim~\ref{cla:good-case}.
Hence, we show that $\|\Phi(\overline{\bfR}(x))\|_{U^d}$ is also small in the following two claims.
\begin{claim}\label{cla:gowers-norm-of-phi-overline-bfR'-small}
  \[
    \|\Phi(\overline{\bfR'}(x))\|_{U^d} \leq \gamma + \delta^{1/2^d} + 2\rho.
  \]
\end{claim}
\begin{proof}
  Recall that $\Phi(\overline{\bfR'}(x)) = f_2(Ax) + f_3(Ax) - \Upsilon'(x) - \Delta'(x)$.
  Hence, 
  \begin{align*}
    \|\Phi(\overline{\bfR'}(x))\|_{U^d} 
    & \leq \|f_2\circ A\|_{U^d} + \|f_3 \circ A\|_{U^d} + \|\Upsilon'\|_{U^d} + \|\Delta'\|_{U^d} \\
    & \leq \rho + \rho + \|\Upsilon'\|_{U^d} + \|\Delta'\|_1^{1/2^d} \quad \text{(from Claim~\ref{cla:good-case})}\\
    & \leq \gamma + \delta^{1/2^d} + 2\rho. 
  \end{align*}
  \qedhere
\end{proof}

\begin{claim}\label{cla:diff-of-gowers-norms-small}
  Suppose $r_{\ref{lem:regularity-instance-soundness}}(\gamma,\epsilon,C,d) \geq r_{\ref{lem:Gamma-decides-gowers-norm}}(\rho,|\overline{\bfR}|,d)$.
  Then,
  \[
    \Bigl|\|\Phi(\overline{\bfR'}(x))\|_{U^d} - \|\Phi(\overline{\bfR}(x))\|_{U^d}\Bigr| \leq \rho.
  \]
\end{claim}
\begin{proof}
  From Claim~\ref{cla:property-of-overline-bfR}, 
  $\bfR$ and $\bfR'$ have the same degree less than $d$ and the same depth.
  Also, ranks of them are at least $r_{\ref{lem:Gamma-decides-gowers-norm}}(\rho,|\bfR|,d)$.
  Hence, the claim follows from Lemma~\ref{lem:Gamma-decides-gowers-norm}.
\end{proof}
In what follows, we assume that $r_{\ref{lem:regularity-instance-soundness}}(\gamma,\epsilon,C,d) \geq r_{\ref{lem:Gamma-decides-gowers-norm}}(\rho,|\bfR|,d)$.
Note that $\rho$ and $|\bfR|$ are functions of $\gamma$, $\epsilon$, $C$, and $d$.
From Claims~\ref{cla:good-case},~\ref{cla:gowers-norm-of-phi-overline-bfR'-small}, and~\ref{cla:diff-of-gowers-norms-small},
we have 
\begin{align*}
  \|f - \Gamma(P)\|_{U^d} & \leq \|f_2\|_{U^d} + \|f_3\|_{U^d} + \|\Phi(\overline{\bfR})\|_{U^d} \\
  & \leq
  \|f_2\|_{U^d} + \|f_3\|_{U^d} + \|\Phi(\overline{\bfR'})\|_{U^d} + \Bigl| \|\Phi(\overline{\bfR'})\|_{U^d} - \|\Phi(\overline{\bfR})\|_{U^d}\Bigr| \\
  & \leq \rho + \rho + \gamma + \delta^{1/2^d} + 2\rho + \rho = \gamma + \delta^{1/2^d} + 5\rho.
\end{align*}
By setting $\delta$ and $\rho$ so that $\delta^{1/2^d} + 5\rho \leq \tau_{\ref{lem:small-perturbation}}(\gamma,\epsilon,d)$ and $r_{\ref{lem:regularity-instance-soundness}}(\gamma,\epsilon,C,d) \geq r_{\ref{lem:small-perturbation}}(\gamma,\epsilon,C,d)$,
the function $f$ is $\epsilon$-close to satisfying $I$ from Lemma~\ref{lem:small-perturbation}.
We reach a contradiction, and Lemma~\ref{lem:regularity-instance-soundness} follows.

\section{Any Locally Testable Property is Regular-Reducible}\label{sec:testable->regular-reducible}
In this section, we prove Theorem~\ref{the:testable->regular-reducible}.

Consider a function $f :\Fp^n \to \{0,1\}$, and an integer $m \leq n$.
Let $\mu_{f,m}$ denote the distribution of $f \circ A : \Fp^m \to \bit$,
where $A:\Fp^m \to \Fp^n$ is over random affine embeddings.
For $v: \Fp^m \to \bit$, we denote by $\mu_{f,m}[v]$ the probability that $f \circ A$ coincides with $v$.

These notions can be generalized to functions $f: \Fp^n \to [0,1]$.
We view such functions as distribution over functions $f': \Fp^n \to \bit$,
where $\Pr[f'(x) = 1] = f(x)$ independently for all $x \in \Fp^n$.
Let again $A: \Fp^m \to \Fp^n$ be a random affine embedding,
and we denote by $\mu_{f,m}$ the distribution of $f' \circ A: \Fp^m \to \bit$.
This is a generalization of the former case as a function $f: \Fp^n \to \bit$ can be identified with the function that maps every $x \in \Fp^n$ to the point-mass probability distribution over $\bit$ which is concentrated on $f(x)$.
The following lemma says that,
in order to show that the statistical distance between $\mu_{f,m}$ and $\mu_{g,m}$ is small,
it suffices to show that the Gowers norm of $f - g$ is small.
\begin{lemma}[Lemma 3.3 of~\cite{Hatami:2013ux}]\label{lem:small-gowers-norm->small-statistical-distance}
  For every $\epsilon > 0$ and $m \in \bbN$,
  there exist $\rho = \rho_{\ref{lem:small-gowers-norm->small-statistical-distance}}(\epsilon,m)$ and $d = d_{\ref{lem:small-gowers-norm->small-statistical-distance}}(\epsilon,m)$ with the following property.
  For any functions $f, g:\Fp^n \to [0,1]$ with $\|f - g\|_{U^d}\leq \rho$,
  we have $\dtv(\mu_{f,m},\mu_{g,m}) \leq \epsilon$.
\end{lemma}

Suppose we have a function $f(x) = \Gamma(\bfP(x))$ such that the rank of $\bfP$ is high.
Then the following lemma says that the distribution of $\mu_{f,m}$ is determined by the function $\Gamma$ and the degree and the depth of $\bfP$, and not by specific form of $\bfP$.
\begin{lemma}[Lemma 3.5 of~\cite{Hatami:2013ux}]\label{lem:only-gamma-matters}
  For any $\epsilon > 0$, and $d, m \in \bbN$,
  there exists $r = r_{\ref{lem:only-gamma-matters}}^{(\epsilon,d,m)}: \bbN \to \bbN$ with the following property.
  Let $\bfP$ and $\bfQ$ be polynomial sequences (possibly on different numbers of variables) with the same complexity $C$, the same degree at most $d$, the same depth $\bfh = (h_1,\ldots,h_C)$, and ranks at least $r(C)$.
  Let $\Gamma: \prod_{i\in [C]} \bbU_{h_i+1} \to [0,1]$ be a function.
  Then, $\dtv(\mu_{\Gamma \circ \bfP,m}, \mu_{\Gamma \circ \bfQ,m}) \leq \epsilon$ holds.
\end{lemma}

We can show a similar lemma if,
instead of replacing $\bfP$ by $\bfQ$,
we replace $\Gamma$ with a similar structure function.
\begin{lemma}\label{lem:even-similar-gamma-is-ok}
  For any $\epsilon > 0$, and $m \in \bbN$,
  there exist $d = d_{\ref{lem:even-similar-gamma-is-ok}}(\epsilon,m)$ and $\tau = \tau_{\ref{lem:even-similar-gamma-is-ok}}(\epsilon,m)$ with the following property.
  Let $\bfP$ be a polynomial sequence with complexity $C$ and degree at most $d$.
  Let $\Gamma,\widetilde{\Gamma}: \bbT^C \to [0,1]$ be functions with $\|\Gamma - \widetilde{\Gamma}\|_\infty \leq \tau$.
  Then, $\dtv(\mu_{\Gamma \circ \bfP,m}, \mu_{\widetilde{\Gamma} \circ \bfP,m}) \leq \epsilon$ holds.
\end{lemma}
\begin{proof}
  We set $\rho = \rho_{\ref{lem:small-gowers-norm->small-statistical-distance}}(\epsilon,m)$, $d_{\ref{lem:even-similar-gamma-is-ok}}(\epsilon,m) = d_{\ref{lem:small-gowers-norm->small-statistical-distance}}(\epsilon,m)$, and $\tau_{\ref{lem:even-similar-gamma-is-ok}}(\epsilon,m) = \rho^{1/2^d}$.
  We have
  \[
    \|\Gamma \circ \bfP - \widetilde{\Gamma} \circ \bfP\|_{U^d}^{2^d}
    = \left|\E_{x,y_1,\ldots,y_d}\prod_{I \subseteq [d]} ((\Gamma\circ \bfP)(x + \sum_{i \in I}y_i) - (\widetilde{\Gamma}\circ \bfP)(x + \sum_{i \in I}y_i)) \right|
    \leq \E_{x,y_1,\ldots,y_d} \tau^{2^d} = \tau^{2^d}.
  \]
  Hence $\|\Gamma \circ \bfP - \widetilde{\Gamma} \circ \bfP\|_{U^d} \leq \rho$ holds.
  From Lemma~\ref{lem:small-gowers-norm->small-statistical-distance},
  we have $\dtv(\mu_{\Gamma \circ \bfP,m}, \mu_{\widetilde{\Gamma} \circ \bfP,m}) \leq \epsilon$.
\end{proof}

Combining Lemmas~\ref{lem:only-gamma-matters} and~\ref{lem:even-similar-gamma-is-ok},
we have the following.
\begin{corollary}\label{cor:only-gamma-matters-combined}
  For any $\epsilon > 0$ and $m \in \bbN$,
  there exist $\tau = \tau_{\ref{cor:only-gamma-matters-combined}}(\epsilon,m)$,
  $r = r_{\ref{cor:only-gamma-matters-combined}}^{(\epsilon,m)}: \bbN \to \bbN$, and 
  $d = d_{\ref{cor:only-gamma-matters-combined}}(\epsilon,m)$ with the following property.
  Let $\bfP$ and $\bfQ$ be polynomial sequences (possibly on different numbers of variables) with the same complexity $C$, the same degree at most $d$, the same depth, and ranks at least $r(C)$.
  Let $\Gamma,\widetilde{\Gamma}: \bbT^C \to [0,1]$ be functions with $\|\Gamma - \widetilde{\Gamma}\|_\infty \leq \tau$.
  Then we have $\dtv(\mu_{\Gamma \circ \bfP,m}, \mu_{\widetilde{\Gamma} \circ \bfQ,m}) \leq \epsilon$.  
\end{corollary}



The lemma above motivates us to define a typical distribution obtained from functions satisfying a regularity-instance as follows.
\begin{definition}
  Let $I = (\gamma,\Gamma,C,d,\bfd,\bfh,r)$ be a regularity-instance.
  Then, we define a distribution $\mu_{I,m}$ over $\Fp^m \to \bit$ as
  \begin{align*}
    \mu_{I,m}[v] = \mu_{\Gamma \circ \bfP,m}[v]
  \end{align*}
  for each $v: \Fp^m \to \bit$, where $\bfP$ is an arbitrary polynomial sequence with complexity $C$, $\deg(\bfP) = \bfd$, $\depth(\bfP) = \bfh$, and $\rank(\bfP) \geq r$.
\end{definition}

\begin{proof}[Proof of Theorem~\ref{the:testable->regular-reducible}]
  Suppose $\delta < 1/6$ (if $\delta \geq 1/6$, we set it to, say $1/10$).
  Let $m = m(\delta)$ and $\calV = \calV(\delta) \subseteq \{\Fp^m \to \bit\}$ 
  given by Proposition~\ref{pro:canonical-tester} with the proximity parameter $\delta$.
  We first set $d = \max\{d_{\ref{lem:small-gowers-norm->small-statistical-distance}}(\delta/2,m),d_{\ref{cor:only-gamma-matters-combined}}(\delta/2,m)\}$.
  Then, we choose $\gamma = \gamma_{\ref{lem:small-gowers-norm->small-statistical-distance}}(\delta/2,m)$,
  $\tau = \min\{\tau_{\ref{lem:small-perturbation}}(\gamma,\delta,d),\tau_{\ref{cor:only-gamma-matters-combined}}(\delta/2,m)\}$,
  $\rho = \rho_{\ref{lem:small-gowers-norm->small-statistical-distance}}(\delta/2,m)$.
  We also choose $\zeta \leq (\rho/2)^{2^d}$ and $\eta:\bbN \to \bbN$ as $\eta(D) \leq \rho/2$ for any $D \in \bbN$.
  Finally, we define $r:\bbN \to \bbN$ as $r(D) \geq \max\{r_{\ref{the:regularity-instance->testable}}(\gamma,\delta/8,D,d),$ $ r_{\ref{lem:small-perturbation}}(\gamma,\delta,D,d),$ $ r_{\ref{cor:only-gamma-matters-combined}}^{(\delta/2,m)}(D)\}$ for any $D \in \bbN$, and $\overline{C} = C_{\ref{the:regularity-lemma}}(\eta, \zeta, 0, d, r)$.

  For any $1 \leq C \leq \overline{C}$,
  consider all of the (finitely many) regularity-instances $I$ such that each value of the structure function of $I$ is a multiple of $\tau$,
  the degree-bound parameter is $d$,
  the rank parameter is at least $r(C)$ for its complexity parameter $C$,
  and its complexity is at most $\max(1/\gamma,\overline{C}, d, r(\overline{C}))$.
  Let $\calR$ be the union of all these regularity-instances.
  Note that, all the above constants, as well as the size of $\calR$ are determined as a function of $\delta$ only (and the property $\calP$).
  Also from the choice of $r$, every instance in $\calR$ is of high rank with respect to $\delta$.

  We claim that we can take $\calI$ in Definition~\ref{def:regular-reducible} to be
  \begin{align*}
    \calI = \left\{I \in \calR : \sum_{v \in \calV}\mu_{I,m}[v] \geq \frac 12 \right\}.
  \end{align*}

  Suppose that a function $f$ satisfies $\calP$.
  We decompose $f$ as $f = f_1 + f_2+f_3$ using Lemma~\ref{the:regularity-lemma} with parameters $\eta$, $\zeta$, $0$ (corresponding to $C_0$), $d$, and $r$.
  Note that $f_1 $ can be expressed as $\Gamma \circ \bfP$,
  where $\bfP = (P_1,\ldots,P_C)$ is a polynomial sequence with $C \leq \overline{C}$, degrees less than $d$, and $\rank(\bfP) \geq r(C)$,
  and $\Gamma:\prod_{i \in [C]}\bbU_{\depth(P_i)+1} \to [0,1]$ is a function.

  By the construction of $\calR$,
  some regularity-instance $I \in \calR$ has a structure function $\Gamma_I$ with $\|\Gamma_I - \Gamma\|_\infty \leq \tau$.
  From Lemma~\ref{cor:only-gamma-matters-combined},
  we have $\dtv(\mu_{f_1,m}, \mu_{I,m}) \leq \delta/2$.
  Also from the choice of $\eta$ and $\zeta$,
  \begin{align*}
    \|f - f_1\|_{U^d}
    \leq
    \|f_2\|_{U^d} + \|f_3\|_{U^d} 
    \leq 
    \|f_2\|_{2}^{1/2^d} + \|f_3\|_{U^d} 
    \leq
    \zeta^{1/2^d} + \eta(C)
    \leq
    \rho/2 + \rho/2 = 
    \rho.
  \end{align*}
  From the choice of $\rho$ and $d$,
  by Lemma~\ref{lem:small-gowers-norm->small-statistical-distance}, 
  we have $\dtv(\mu_{f,m},\mu_{f_1,m}) \leq \delta/2$.
  It follows that $\dtv(\mu_{f,m},\mu_{I,m}) \leq \delta$.
  Recall that Lemma~\ref{pro:canonical-tester} indicates that $\sum_{v \in \calV}\mu_{f,m}[v] \geq \frac{2}{3}$.
  Hence, $\sum_{v \in \calV}\mu_{I,m}[v] \geq 2/3 - \delta/2 \geq 1/2$ (here we use $\delta < 1/6$), and we have $I \in \calI$.
  Also from Lemma~\ref{lem:small-perturbation} and the choice of $\tau$,
  we have $f$ is $\delta$-close to satisfying $I$.
  Hence, $f$ is indeed $\delta$-close to satisfying one of regularity-instances in $\calI$.

  Suppose now that a function $f$ is $\epsilon$-far from satisfying $\calP$.
  If $\delta \geq \epsilon$, then there is nothing to prove.
  So assume that $\delta < \epsilon$.
  If $f$ is $(\epsilon-\delta)$-close to satisfying a regularity-instance $I = (\gamma,\Gamma,C,d,\bfd,\bfh,r)\in \calI$.
  Then, there exists a polynomial sequence $\bfP$ with $\deg(\bfP) = \bfd$ and $\depth(\bfP) = \bfh$, and $\rank(\bfP) \geq r$ such that $f$ is $(\epsilon-\delta)$-close to a function $g(x) = \Gamma(\bfP(x)) + \Upsilon(x)$ with $\|\Upsilon\|_{U^d} \leq \gamma$.
  From the choice of $\gamma$ and $d$,
  by Lemma~\ref{lem:small-gowers-norm->small-statistical-distance},
  we have $\dtv(\mu_{g,m},\mu_{\Gamma\circ \bfP,m}) \leq \delta/2$.
  Also, we have $\dtv(\mu_{\Gamma \circ \bfP},\mu_I) \leq \delta/2$ from Lemma~\ref{cor:only-gamma-matters-combined}.
  Hence $\sum_{v \in \calV}\mu_{g,m}[v] \geq 1/2-\delta > 1/3$,
  and it follows that the tester accepts $g$ with probability more than $1/3$.
  This implies that $g$ is $\delta$-close to satisfying $\calP$.
  However, since $f$ is $\epsilon$-far from satisfying $\calP$,
  any function that is $(\epsilon-\delta)$-close to $f$ must be $\delta$-far from satisfying $\calP$,
  a contradiction.
\end{proof}

\section{Any Regular-Reducible Property is Locally Testable}\label{sec:reducible->testable}
In this section, we prove Theorem~\ref{the:regular-reducible->testable} using Theorem~\ref{the:regularity-instance->testable}.
Let us introduce the concept of tolerant testers.
\begin{definition}[Tolerant testers]\label{def:tolerant-tester}
  An algorithm is called an $(\epsilon_1,\epsilon_2)$-tester for a property $\calP$ if,
  given a query access to a function $f:\Fp^n \to \bit$,
  with probability at least $2/3$,
  it accepts when $f$ is $\epsilon_1$-close to $\calP$,
  and rejects when $f$ is $\epsilon_2$-far from $\calP$.
\end{definition}

The following theorem says that,
if a property is locally testable,
then we can estimate the distance to the property with a constant query complexity.
\begin{theorem}[\cite{Hatami:2013ux}]\label{the:testable->estimable}
  Let $\calP$ be an affine-invariant locally testable property.
  Then, for every $0 \leq \epsilon_1 < \epsilon_2 \leq 1$,
  there is an $(\epsilon_1,\epsilon_2)$-tester for $\calP$ whose query complexity only depends on $\epsilon_2 - \epsilon_1$ (and $\calP$),
  which is independent of the input size.
\end{theorem}

We want to apply Theorem~\ref{the:testable->estimable} to the property of satisfying regularity-instances.
For a regularity-instance $I = (\gamma,\Gamma,C,d,\bfd,\bfh,r)$,
let $\calP_I$ be the property of satisfying $I$.
An issue here is that $\calP_I$ is locally testable with error parameter $\epsilon$ only when $r \geq r_{\ref{the:regularity-instance->testable}}(\gamma,\epsilon,C,d)$.
Hence, we cannot simply say that $\calP_I$ is locally testable regardless of $\epsilon$,
and apparently we cannot apply Theorem~\ref{the:testable->estimable}.
However, closely looking at the proof of Theorem~\ref{the:testable->estimable},
to estimate the distance to a property with parameters $\epsilon_1$ and $\epsilon_2$,
we only need that the property is locally testable with a proximity parameter $(\epsilon_2-\epsilon_1)/8$.
Hence, we have the following corollary.
\begin{corollary}\label{cor:regularity-instance-estimable}
  Let $\gamma > 0$,
  $0 \leq \epsilon_1 < \epsilon_2 \leq 1$, and $C, d \in \bbN$.
  For any regularity-instance $I = (\gamma,\Gamma,C,d,\bfd,\bfh,r)$ with rank at least $r_{\ref{the:regularity-instance->testable}}(\gamma,(\epsilon_2-\epsilon_1)/8,C,d)$,
  there is an $(\epsilon_1,\epsilon_2)$-tester for the property of satisfying $I$ with query complexity that depends only on $\epsilon_2 - \epsilon_1$ and $I$.
\end{corollary}

\begin{proof}[Proof of Theorem~\ref{the:regular-reducible->testable}]
  Suppose that a property $\calP$ is regular-reducible as per Definition~\ref{def:regular-reducible}.
  Let us fix $n$ and $\epsilon$.
  Put $s = s(\epsilon/4)$, and let $\calI$ be the corresponding set of regularity-instances for $\delta = \epsilon/4$ as in Definition~\ref{def:regular-reducible}.
  Recall that Definition~\ref{def:regular-reducible} guarantees that the number and the complexity of the regularity-instances in $\calI$ are bounded by a function of $\delta$.

  Since each regularity-instance in $\calI$ is of high rank with respect to $\delta/8$ and $\delta \leq (\epsilon - \delta) - \delta$,
  by Theorem~\ref{the:testable->estimable}, for any such $I$,
  there is a $(\delta,\epsilon-\delta)$-tester for the property of satisfying $I$ with query complexity that depends only on $\epsilon$ (and $I$).
  In particular, by repeating the algorithm of Theorem~\ref{the:testable->estimable} an appropriate number of times (that depends only on $s$),
  and taking the majority vote, 
  we get an algorithm for distinguishing between the above two cases,
  whose query complexity is a function of $\epsilon$ and $s$,
  which succeeds with probability of at least $1-1/3s$.
  As $s$ itself is bounded by a function of $\epsilon$,
  the number of queries of this algorithm is bounded by a function of $\epsilon$ only.

  We are now ready to describe our tester for $\calP$:
  Given a function $f:\Fp^n \to \bit$ and $\epsilon > 0$,
  for every $I \in \calI$,
  the algorithm uses the version of Theorem~\ref{the:testable->estimable} described in the previous paragraph,
  which succeeds with probability at least $1- 1/3s$ in distinguishing between the case that $f$ is $\delta$-close to satisfying $I$ from the case that it is $(\epsilon-\delta)$-far from satisfying it.
  If it finds that $f$ is $\delta$-close to satisfying some $I \in \calI$,
  then it accepts $f$;
  otherwise it rejects $f$.

  Observe that, as there are at most $s$ regularity-instances in $\calI$, 
  we get by the union bound that with probability at least $2/3$,
  the subroutine for estimating how far is $f$ from satisfying some $I \in \calI$ never errs.
  We now prove that the above algorithm is indeed a tester for $\calP$. 
  Suppose first that $f$ satisfies $\calP$.
  As we set $\delta = \epsilon/4$ and $\calP$ is regular-reducible to $\calI$, 
  the function $f$ must be $\delta$-close to satisfying some regularity-instance $I \in \calI$. 
  Suppose now that $G$ is $\epsilon$-far from satisfying $\calP$.
  Again, as we assume that $\calP$ is regular-reducible to $\calI$, 
  we conclude that $f$ must be $(\epsilon-\delta)$-far from satisfying all of the regularity-instances $I \in \calI$.
  We get that if $f$ satisfies $\calP$, then with probability at least $2/3$,
  the algorithm will find that $f$ is $\delta$-close to satisfying some $I \in \calI$, 
  while if $f$ is $\epsilon$-far from satisfying $\calP$,
  then with probability at least $2/3$,
  the algorithm will find that $f$ is $(\epsilon-\delta)$-far from all $I \in \calI$.
  By the definition of the algorithm, we get that with probability at least $2/3$ it distinguishes between functions satisfying $\calP$ from those that are $\epsilon$-far from satisfying $\calP$. 
  This means that the algorithm is indeed an $\epsilon$-tester for $\calP$.
\end{proof}

\section{Rank-Oblivious Regular-Reducibility}\label{sec:rank-oblivious-reducibility}
In this section, we prove Theorem~\ref{the:rank-obliviously-reducible->reducible}, and then apply it to show that the property of being a classical low-degree polynomial is locally testable.
\begin{proof}[Proof of Theorem~\ref{the:rank-obliviously-reducible->reducible}]
  Fix $\delta > 0$ and $n \in \bbN$.
  Set $\delta' = \delta / 3$.
  Let $\calI' = \calI'(\delta')$ be the set of rank-oblivious regularity-instances as in Definition~\ref{def:rank-obliviously-regular-reducible} with the parameter $\delta'$.
  The complexity of $\calI'$ is bounded from above by $s' = s'(\delta')$.

  We now describe how to construct a set of regularity-instances $\calI$ for Definition~\ref{def:regular-reducible}.
  Let $\rho = \rho_{\ref{lem:small-gowers-norm->small-statistical-distance}}(\delta',1)$ and $d = d_{\ref{lem:small-gowers-norm->small-statistical-distance}}(\delta',1)$.
  For $\gamma = \rho/3$,
  we set $r:\bbN \to \bbN$ so that $r(D) = \max\{r_{\ref{the:regularity-instance->testable}}(\delta/8,D,d), $ $
  r_{\ref{lem:Gamma-decides-gowers-norm}}(\gamma,D,d) \}$ for any $D \in \bbN$.

  For any function $f:\Fp^n \to \bit$ satisfying $\calP$,
  there exists a rank-oblivious regularity-instance $I_L = (\gamma_L,\Gamma_L,C_L,d_L,\bfd_L,\bfh_L) \in \calI_L$ such that $f$ is $\delta'$-close to satisfying $I_L$ ($L$ stands for ``low rank'').
  Hence, $f(x)$ is $\delta'$-close to a function $f_L:\Fp^n \to \bit$ of the form $f_L(x) = \Gamma_L(\bfP_L(x)) + \Upsilon_L(x)$ for a polynomial sequence $\bfP_L$ with degree $\bfd_L$ and depth $\bfh_L$, and a function $\Upsilon_L$ with $\|\Upsilon_L\|_{U^{d_L}} \leq \gamma_L$.
  Let $\zeta_L = (\gamma/2)^{2^d}$ and $\eta_L:\bbN \to \bbN$ so that $\eta_L(D) \leq \gamma/2$ for every $D \in \bbN$.
  Then using Lemma~\ref{the:regularity-lemma},
  we have that 
  $f_L(x) = f_{L,1}(x) + f_{L,2}(x) + f_{L,3}(x)$ such that
  $f_{L,1}(x) = \Gamma_H(\bfP_H)$ for a polynomial sequence $\bfP_H$ of complexity $C_H \leq C_{\ref{lem:polynomial-regularity-lemma}}(\eta_L,\zeta_L,C_L,d,r)$,
  degree less than $d$, and rank at least $r$, and a function $\Gamma_H:\bbT^{C_H} \to [0,1]$ ($H$ stands for ``high rank'').
  Also, $\|f_{L,2}\|_2 \leq \zeta_H$ and $\|f_{L,3}\|_{U^{d}} \leq \eta(C_H)$.
  We note that $\|f_{L,2} + f_{L,3}\|_{U^d} \leq \zeta_H^{1/2^d} + \eta(C_H) \leq \gamma$.
  We add to $\calI$ the regularity-instance $I_H = (\gamma,\Gamma_H,C_H,d,\deg(\bfP_H),\depth(\bfP_H),\rank(\bfP_H))$.

  Now we see that $\calI$ satisfies the condition of Definition~\ref{def:regular-reducible}.
  First, the complexity of any regularity-instance in $\calI$ is bounded from above by a function of $\delta$.
  Also, any regularity-instance in $\calI$ is of high rank with respect to $\delta$.

  If a function $f$ satisfies $\calP$,
  then there is some $I_H \in \calI$ such that $f$ is $\delta'$-close to satisfying $I_H$ from the construction of $\calI$.
  Since $\delta' \leq \delta$,
  the function $f$ is $\delta$-close to satisfying $I_H$.

  Suppose that a function $f$ is $\epsilon$-far from satisfying $\calP$.
  Assume that $f$ is $(\epsilon-\delta)$-close to satisfying some regularity-instance $I_H = (\gamma,\Gamma_H,C_H,d,\bfd_H,\bfh_H,r_H) \in \calI$.
  Then, $f$ is $(\epsilon-\delta)$-close to a function $f_H$ of the form $f_H(x) = \Gamma_H(\bfP_H(x)) + \Upsilon_H(x)$ for a polynomial sequence $\bfP_H$ with complexity $C_H$, $\deg(\bfP_H) = \bfd_H$, $\depth(\bfP_H) = \bfh_H$, and $\rank(\bfP_H) \geq r_H \geq r$, and a function $\Upsilon_H:\Fp^n \to [-1,1]$ with $\|\Upsilon_H\|_{U^{d}} \leq \gamma$.
  From the construction of $\calI$, there is a function $g$ satisfying $\calP$ such that $g$ is $\delta'$-close to a function $g_H$ of the form $g_H = \Gamma_H(\bfQ_H(x)) + \Upsilon'_H(x)$ for a polynomial sequence $\bfQ_H$ with complexity $C_H$, $\deg(\bfQ_H) = \bfd_H$, $\depth(\bfQ_H) = \bfh_H$, and $\rank(\bfQ_H) \geq r$, and a function $\Upsilon'_H:\Fp^n \to [-1,1]$ with $\|\Upsilon'_H\|_{U^d} \leq \gamma$.
  Note that $\|f_H - g_H\|_{U^d} \leq \|\Gamma \circ \bfP_H - \Gamma \circ \bfQ_H\|_{U^d} + \|\Upsilon_H\|_{U^d} + \|\Upsilon'_H\|_{U^d}$.
  From Lemma~\ref{lem:Gamma-decides-gowers-norm},
  we have $\|\Gamma \circ \bfP_H - \Gamma \circ \bfQ_H\|_{U^d} \leq \gamma$.
  Hence, $\|f_H - g_H\|_{U^d} \leq 3\gamma = \rho$.
  From Lemma~\ref{lem:small-gowers-norm->small-statistical-distance},
  we have $\dtv(\mu_{f_H,1},\mu_{g_H,1}) \leq \delta'$.
  In particular, $\|f_H - g_H\|_1 \leq \delta'$ holds.
  This means that the distance between $f$ and $g$ is at most $\|f - f_H\|_1 + \|f_H - g_H\|_1 + \|g_H - g\|_1 \leq  \epsilon-\delta+\delta'+\delta' = \epsilon - \delta/3$,
  which is contradicting that $f$ is $\epsilon$-far from $\calP$.
  Hence, $f$ is $(\epsilon-\delta)$-far from satisfying any regularity-instance in $\calI$.
\end{proof}

We apply Theorem~\ref{the:rank-obliviously-reducible->reducible} to obtain the following, which is already known~\cite{Alon:2005jl}.
\begin{corollary}
  For any $d \in \bbN$,
  the property of being a classical degree-$d$ polynomial is locally testable.
\end{corollary}
\begin{proof}
  We show that the property is rank-obliviously regular-reducible.
  Fix $\delta > 0$ and $n \in \bbN$.
  We set $\rho = \rho_{\ref{lem:small-gowers-norm->small-statistical-distance}}(\delta/2,1)$ and $\overline{d} = d_{\ref{lem:small-gowers-norm->small-statistical-distance}}(\delta/2,1)$.
  For each $k \in \{0,\ldots,d-1\}$,
  we define a rank-oblivious regularity-instance $I_k = (\rho,\mathrm{id},1,\overline{d},(k),(0))$.
  Then, we choose $\calI = \{I_0,\ldots,I_d\}$ as the set of rank-oblivious regularity-instances to which we reduce $\calP$.
  Now we check that $\calI$ satisfies the condition of Definition~\ref{def:rank-obliviously-regular-reducible}.

  Suppose that a function $f:\Fp^n \to \bit$ is a classical polynomial of degree $k < d$.
  Then, it is clear that $f$ satisfies the rank-oblivious regularity-instance $I_k$.

  Suppose that $f:\Fp^n \to \bit$ is $\epsilon$-far from classical polynomials of degree less than $d$.
  Assume for contradiction that $f$ is $(\epsilon-\delta)$-close to a rank-oblivious regularity-instance $I_k$ for some $k \in \{0,\ldots,d\}$.
  Then, $f$ is $(\epsilon-\delta)$-close to a function $f':\Fp^n \to \bit$ of the form $f'(x) = P(x) + \Upsilon(x)$,
  where $P:\Fp^n \to \bit$ is a classical polynomial of degree $k$ and $\Upsilon:\Fp^n \to [-1,1]$ is a function with $\|\Upsilon\|_{U^{\overline{d}}} \leq \rho$.
  From Lemma~\ref{lem:small-gowers-norm->small-statistical-distance},
  the distance between $f'$ and $P(x)$ is at most $\delta/2$.
  This implies the distance between $f$ and $P$ is $\epsilon - \delta + \delta/ 2 = \epsilon - \delta/2 < \epsilon$,
  which is a contradiction.
\end{proof}

\section*{Acknowledgments}
The author would like to thank Arnab Bhattacharyya for valuable discussions.

\end{document}